\documentclass[aps,prl,twocolumn,superscriptaddress,floatfix,longbibliography]{revtex4-2}
\usepackage{amsmath}
\usepackage{appendix}
\usepackage{times}
\usepackage{amssymb}
\usepackage{graphicx}
\usepackage{soul}
\usepackage[colorlinks=true,linkcolor=blue, citecolor=blue, urlcolor=blue, bookmarks]{hyperref}
\makeatletter
\usepackage[T1]{fontenc}
\usepackage{amsmath,amsfonts,amssymb,amsthm,bbm,bm, braket}
\usepackage{wasysym}
\usepackage{comment}
\usepackage{scalerel}
\usepackage{enumerate}
\usepackage{amssymb}
\usepackage{graphicx}
\usepackage{mathtools}
\usepackage{ifthen}
\usepackage{tensor}
\usepackage{tikz}
\usepackage{tikz-network}
\usetikzlibrary{patterns,decorations.pathreplacing}

\usepackage{color}
\usepackage{longtable}
\usepackage[normalem]{ulem} 
\theoremstyle{break}        
\usepackage{color}
\definecolor{myred}{RGB}{232,102,102}
\definecolor{myblue}{RGB}{187,187,255}
\definecolor{myorange0}{RGB}{252,226,5}
\definecolor{myorange0c}{RGB}{255,255,255}
\definecolor{myorange}{RGB}{255,165,0}
\definecolor{mygrey}{RGB}{105,105,105}
\definecolor{OliveGreen}{RGB}{85,107,47}
\definecolor{NavyBlue}{RGB}{0,0,128}
\definecolor{mygreen}{RGB}{34,139,34}
\definecolor{myY}{RGB}{220,255,203}
\definecolor{myYO}{RGB}{255, 220, 151}

\definecolor{mygreenc}{RGB}{150,50,50}

\usepackage{tensor}
\newcommand{\be}{\begin{equation}}
\newcommand{\ee}{\end{equation}}
\newcommand{\ba}{\begin{aligned}}
\newcommand{\ea}{\end{aligned}}
\newcommand{\bw}{\begin{widetext}}
\newcommand{\ew}{\end{widetext}}

\newtheorem{theorem}{Theorem}
\theoremstyle{plain}

\theoremstyle{plain}
\newtheorem{lemma}{Lemma}
\theoremstyle{plain}

\theoremstyle{plain}

\renewcommand{\qedsymbol}{$\blacksquare$}

\newcommand{\Wgategreenprime}[2]{
\draw[very thick] (#1-0.5, #2 +0.5) -- (#1+0.5,#2-0.5);
\draw[very thick] (#1-0.5,#2-0.5) -- (#1+0.5,#2+0.5);
\draw[ thick, fill=mygreen!70, rounded corners=2pt] (#1-0.35,#2+0.35) rectangle (#1+0.35,#2-0.35);
\Text[x=#1,y=#2]{$k'$}
}

\newcommand{\Wgategreen}[2]{
\draw[very thick] (#1-0.5, #2 +0.5) -- (#1+0.5,#2-0.5);
\draw[very thick] (#1-0.5,#2-0.5) -- (#1+0.5,#2+0.5);
\draw[ thick, fill=mygreen!70, rounded corners=2pt] (#1-0.35,#2+0.35) rectangle (#1+0.35,#2-0.35);
\Text[x=#1,y=#2]{$k$}
}

\newcommand{\Measurement}[2]{
\draw[very thick](#1-0.5,#2-0.5)--(#1,#2)--(#1+0.5,#2-0.5);
\draw[fill=orange!30, draw=black, very thick] (#1,#2) circle (0.35);
}
\newcommand{\MYcircle}[2]{
\draw[thick, fill=white] (#1,#2) circle (0.15); }
\newcommand{\MYsquare}[2]{
 \coordinate (Origin) at (#1,#2);
\filldraw [thick, fill=white, even odd rule] ($(Origin)+(-.1cm,-.1cm)$) coordinate (Square) -- ++(0.0cm,0.2cm) -- ++(0.2cm,0.0cm) -- ++(0.0cm,-0.2cm) -- cycle;
 }

\newcommand{\pk}[1]{{\color{blue}[#1]}}
\definecolor{skyblue}{RGB}{135, 206, 235}
\hypersetup{
 pdftitle={Mixed state deep thermalization},
 pdfsubject={Many-body systems},
 pdfauthor={Xie-Hang Yu, Wen Wei Ho, and Pavel Kos},
 pdfkeywords={exact solutions, quantum circuits, quantum many-body systems}
}

\begin{document}
\title[]{Mixed state deep thermalization
}
\author{Xie-Hang Yu}
\affiliation{Max-Planck-Institut f\"ur Quantenoptik, Hans-Kopfermann-Str. 1, 85748 Garching}
\author{Wen Wei Ho}
\affiliation{Department of Physics, National University of Singapore, Singapore 117551}
\affiliation{Centre for Quantum Technologies, National University of Singapore, 3 Science Drive 2, Singapore 117543}
\author{Pavel Kos}
\affiliation{Max-Planck-Institut f\"ur Quantenoptik, Hans-Kopfermann-Str. 1, 85748 Garching}
\begin{abstract}
    We introduce the notion of the \emph{mixed state projected ensemble (MSPE)}, a collection of mixed states describing a local region of a quantum many-body system, conditioned upon  measurements of the complementary region which are incomplete.
    This constitutes a generalization of the pure state projected ensemble in which measurements are assumed ideal and complete, and which has been shown to tend towards  limiting pure state distributions depending only on symmetries of the system, thus representing a new kind of universality in quantum equilibration dubbed {\it deep thermalization}.  We study the MSPE generated by solvable (1+1)d dual-unitary quantum circuit evolution, and identify the limiting   mixed state distributions which emerge at late times depending on the size of the incomplete measurement, which we assume to be lossy, 
    finding that they correspond to certain random density matrix ensembles known in the literature. 
    We also derive the rate of the emergence of such universality.
    Furthermore, we investigate the quantum information properties of the states composing the ensemble, specifically their capacity to teleport quantum information between the ends of the system. 
    The teleportation fidelity is upper bounded by the quantum conditional entropy, which
    we find exhibits a sharp transition from zero to maximal when the number of measurements lost matches of that the number of degrees of freedom to be teleported. 
    Our results initiate the first investigation of deep thermalization for mixed state ensembles, which are relevant for present-day quantum simulation experiments wherein measurements are typically not perfect, and also amount to a physical and natural way of sampling from hitherto abstract random density matrix ensembles. 
\end{abstract}\maketitle

\emph{Introduction---}Understanding how isolated quantum many-body systems thermalize is one of the most important questions driving the study of quantum dynamics. 
The standard framework involves studying the state of a local subsystem $A$ over time --- that is, probe only observables which have support on $A$ --- and ask if it matches that of the Gibbs state, the maximally entropic state subject to conservation laws~\cite{srednicki1994chaos,rigol2008thermalization,nandkishore2015many,Abanin2019Colloquium}. 
In such a formulation, information of the rest of the system $B$ with which $A$ is entangled with, which is assumed to play the role of a bath,
is discarded.

Recent advances in controlling and probing modern-day quantum simulator experiments~\citep{gross2017quantum,choi2023emergent,evered2023high,gross2021quantum,David2022Quantum,Daniel2016Siteresolved,google_2024_quantum,Moses2023Racetrack,liu2025certified} 
have suggested that this framework might, in some cases, be insufficient:
in many such platforms, global measurements can be performed, such that some (classical) information of the bath $B$ might be recorded in addition to the subsystem $A$ of interest.
This has led to a novel perspective on quantum thermalization~\citep{choi2023preparing,cotler2023emergent, ippoliti2023dynamical, ho2022exact}:
the quantum state describing the configuration on $A$ may be studied {\it conditioned} on a  measurement outcome $\bm{\alpha}$ on $B$.
Assuming the measurements are done projectively in a fixed local basis,   considering all conditional states  $\ket{\psi(\bm\alpha)}_{A}$ with measurement outcomes $\bm{\alpha}$  which occur with probability $P_{\bm{\alpha}}$  results in the  so-called \emph{projected ensemble} $\mathcal{E}:=\{P_{\bm{\alpha}},\ket{\psi(\bm\alpha)}_{A}\}$, which can also be understood as a distribution of pure states over the Hilbert space of $A$. 

The projected ensemble subsumes the standard notion of quantum thermalization (its first moment is nothing more than the reduced density matrix), but strictly goes beyond: higher moments encode ergodicity of wavefunctions over the Hilbert space, which is tied to the quantum information scrambling nature of the underlying dynamics~\citep{choi2023preparing,cotler2023emergent, ippoliti2023dynamical, mark2024maximum,zhang2025holographicdeepthermalizationtheory}.

The interest in the projected ensemble lies in the fact that it tends toward limiting distributions (at late times in generic quantum dynamics) which are universal, depending on only coarse grained features of the system such as symmetries and the charge distribution of the initial state. 
This is a novel form of quantum equilibration which has been dubbed \emph{deep thermalization}~\citep{ho2022exact, Ippoliti_2022solvable}.
For example, in the absence of conservation laws, it has been found that the Haar ensemble, a uniform distribution of pure states emerges; while in the presence of conserved quantities, the Scrooge ensemble --- the Haar distribution distorted by the conserved charges --- emerges \cite{mark2024maximum,chang2025deepthermalizationchargeconservingquantum}
~\footnote{More generally, measurement bases may further distort the limiting distributions as they may reveal more or less information about the charges; the resulting distributions are ``generalized Scrooge ensembles''.}.
These have further been argued to be underpinned by generalized maximal entropy principles rooted in quantum information theory, and have   been demonstrated to hold across various physical setting, from spins~\citep{Claeys_2022_emergent, Ippoliti_2022solvable, ippoliti2023dynamical, Harshank2025Nolocality}, 
fermions~\citep{Lucas2023Generalized}, to   continuous variable bosonic systems~\citep{Liu2024GaussianDT}, and even in $1$D and $2$D tensor network states~\citep{lami2025quantumrandomnessemergentconfinement}. 
Importantly, deep thermalization was already experimentally observed in Rydberg quantum simulators~\citep{zhang2025holographicdeepthermalizationtheory}.

However, experimental platforms in reality are never perfect. Many kinds of incoherent errors, like erasure and dephasing in dynamics or measurements, inevitably modify the pure states generated in the projected ensemble into mixed ones, raising questions as to how the universality of deep thermalization will be modified under such scenarios.
In this work, we initiate the investigation of \emph{mixed state deep thermalization}, focusing on the case of a partial loss of measurement outcomes.
We consider the projected ensemble formed under  dynamics of  dual-unitary quantum circuits~\citep{bertini2019exact}, a class of recently-introduced solvable models which has provided powerful and analytical insights into various aspects of quantum dynamics like entanglement growth and quantum chaos~\cite{Akila2016kickedising, bertini2019exact, gopalakrishnan2019unitary,piroli2020exact,claeys2020maximum,rather2020creating,kos2021correlations,bertini2021random,suzuki2022computational}.
We rigorously derive the limiting universal ensembles emerging at late times and identify them as certain random density matrix ensembles known in the random matrix literature~\cite{mehta1991random,zyczkowski2011generating}.
Additionally, we examine the quantum information properties of the conditional states of the projected ensemble, namely their ability to teleport quantum information between distant regions~\citep{Bennett1993teleporting,bao2024finite}.
The teleportation fidelity is bounded from above by the quantum conditional entropy, and we find the latter exhibits a sharp transition from zero to maximal depending on the number of degrees of freedom to be teleported relative to the number of degrees of freedom on which measurement information is lost, signifying a fundamental change in the projected ensemble's quantum information theoretic nature. 
Our results extend the notion of deep thermalization to systems with imperfections, and establishes universality in mixed-state ensembles arising under physical scenarios.

\emph{Mixed state projected ensembles (MSPE)---}Consider a quantum system  in a state {$\ket{\Psi}$}
consisting of two subsystems $A$ and $B$ with $N_A$ and $N_B$ qudits of local dimension $d$, respectively. Subsystem $B$ is further divided into three contiguous sub-parts $B_1,B_2,B_3$ (Fig.~\ref{fig1}) and undergoes rank-$1$ orthogonal measurements, i.e., complete projective measurements. 
However, we assume only the measurement results from $B_1$ and $B_3$ are kept, while those from $m$ sites in $B_2$  are lost. 
These lost measurement outcomes are the particular type of imperfection we consider in our model, leading to mixed states on subsystem $A$ depending on the measurement results of $B_1$ and $B_3$, which we denote by the string vector $\bm{\alpha}=(\bm{\alpha}_{B_1}, \bm{\alpha}_{B_3})$.
For each measurement outcome, there is an associated mixed conditional state on $A$
\begin{equation*}
\rho_{N_{A}}(\bm{\alpha})=\frac{\mathrm{Tr}_{B_2} [(I_{A+B_2}\otimes\bra{\bm{\alpha}})\ket{\Psi}\bra{\Psi}(I_{A+B_2}\otimes\ket{\bm{\alpha}})]
}{P_{\bm{\alpha}}}
,
\end{equation*}
occuring with probability $P_{\bm{\alpha}}=\bra{\Psi}I_{A+B_2}\otimes\ket{\bm{\alpha}}\braket{\bm{\alpha}|\Psi}$. Here $I_{A+B_2}$ is the identity matrix acting on subsystems $A$ and $B_2$, and
$\ket{\bm{\alpha}}=\ket{\bm{\alpha}}_{B_1}\otimes\ket{\bm{\alpha}}_{B_3}$ represents the measurement basis state corresponding to the outcome $\bm{\alpha}$ on $B_1, B_3$. The \emph{mixed state projected ensembles} (MSPE) is the set of all possible mixed states on $A$ combined with their probability $\mathcal{E}:=\{P_{\bm{\alpha}},\rho_{N_A}(\bm{\alpha})$\}.  
Note generalizations of the above set-up can be considered.

\begin{figure}[t]
\includegraphics[width=1.0\columnwidth]{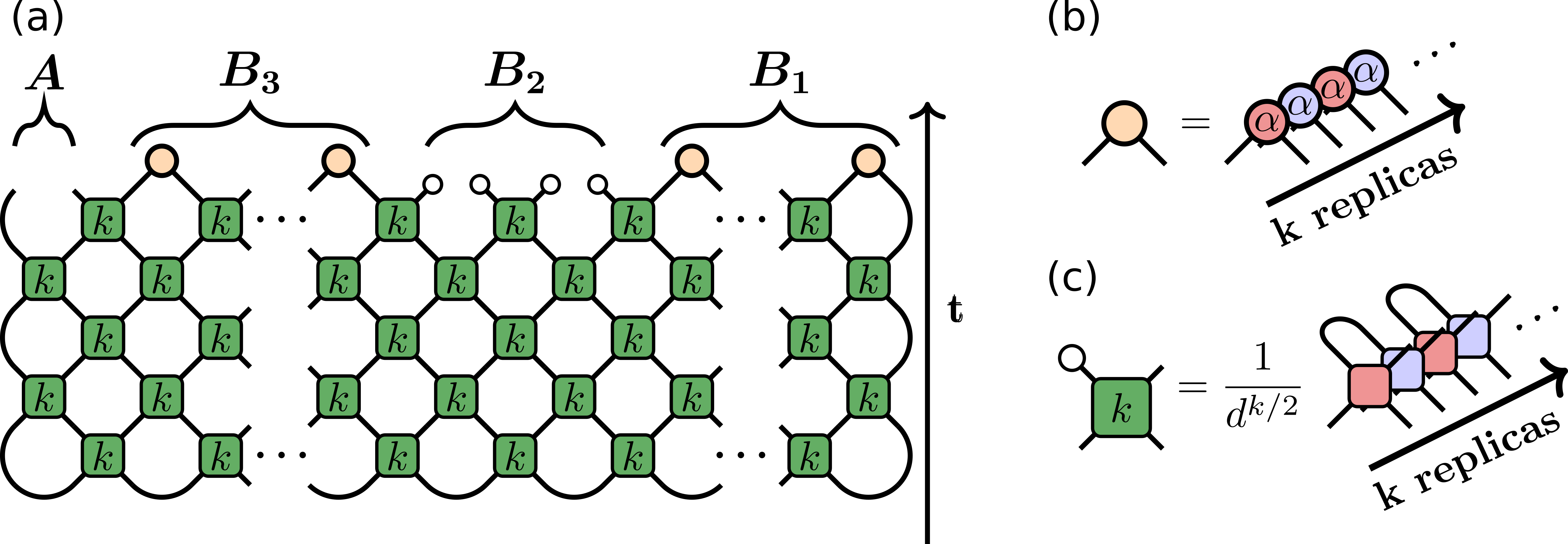}
\caption{
Quench dynamics leading to a mixed state projected ensemble (MSPE). 
(a) The conditional state $\rho_{N_{A}}^{\otimes k}(\bm{\alpha})$ consisting of $k$-pairs of forward and backward evolution ($k$-replica), which together with Born probability $P_{\bm{\alpha}}$ determines the moments of the ensemble $\rho_{N_{A}}^{(k)}$.
A product state composed of maximally entangled pairs evolves under a brick-wall quantum circuit with $t$ layers. The system is bipartited into two subsystems, $A$ and 
$B$, where $B$ is further split into three subparts.
$B$ is measured in $2$-qudit basis, with measurement outcomes from $B_2$ being lost.
We consider $A$'s mixed states conditioned on the measurement outcomes in $B_1, B_3$.
 %
%
(b) The orange circle denotes the projective measurement on two qudits with the measurement outcome $\bm{\alpha}$ in $k$-replicas. 
(c) The white empty bullet corresponds to lost measurement outcomes and is proportional to the $k$-copy of the vectorized identity matrix. This effectively implements a partial trace, connecting the forward and backward time evolutions.
}

\label{fig1}
\end{figure}

The MSPE extends the study of the projected ensembles, which has thus far been focused on pure states, to mixed states, 
and can be understood as a distribution over the space $\mathcal{V}_A$ of density matrices on $A$.
As mentioned, it has been rigorously demonstrated that 
the pure state projected ensemble exhibits universal distributions depending on macroscopic conserved quantities, a phenomenon   known as deep thermalization~\cite{ho2022exact,
ippoliti2023dynamical,Claeys_2022_emergent,mark2024maximum}.
Here, we extend this inquiry to the MSPE, asking whether its distribution converges to a universal one over $\mathcal{V}_A$, that is, one that does not depend on the microscopic details of the model. If so, we call this \emph{mixed-state deep thermalization}.

The distribution of state ensembles can be characterized by their moments. The $k$-th moment is
\begin{equation}
\rho_{N_{A}}^{(k)}=\sum_{\bm{\alpha}}P_{\bm{\alpha}}\rho_{N_{A}}^{\otimes k}(\bm{\alpha}),
\label{eq:moment_definition}
\end{equation}
which is a density matrix acting on the $k$-copy of the subsystem $A$, see Fig.~\ref{fig1} for a graphical expression.
If projected ensemble tends to a universal ensemble $\mathcal{E}'$, $\rho_{N_{A}}^{(k)}$  should match
$\rho_{\mathcal{E}'}^{(k)}=\int\rho^{\otimes k}d\mu(\rho)$ for each  $k$, where $d\mu(\rho)$ is the measure induced by   $\mathcal{E}'$. To quantify the similarity between the two ensembles, we use the Schatten norm
$
\Delta_{\xi}^{(k)}=\frac{\|\rho_{N_{A}}^{(k)}-\rho_{\mathcal{E}'}^{(k)}\|_{\xi}}{\|\rho_{\mathcal{E}'}^{(k)}\|_{\xi}},
$
with $\xi$ the Schatten index.
In this letter, we focus on $\xi=1$ and $2$. 
For $\forall k$, if $\Delta_\xi^{(k)}$ vanishes (is smaller than $\epsilon$), then we say that the MSPE is exactly ($\epsilon$ approximately) deeply thermalized to the ensemble $\mathcal{E}'$. 

\emph{Solvable models and results on $k$-moments---}
We analytically show MSPE arising from a particular treatable protocol, a generalization of the one considered in Ref.~\cite{ippoliti2023dynamical}. 
We take the quantum state of $N_A+N_B$ qudits to
arise from a $1+1$D quantum circuit dynamics, to wit: starting from a product state composed of maximally entangled pairs $\ket{\varphi_0}=\sum_i^d \ket{ii}/\sqrt{d}$ as  $\ket{\psi}= \ket{\varphi_0}^{\otimes (N_A+N_B)/2}$,
the system is evolved with a brick-wall quantum circuit with $t$ layers ($t$ is odd in Fig.~\ref{fig1}~\footnote{If $t$ is even, we can measure the rightmost output qudit in the computational basis, which does not change the results in this letter}),
with each two-qudit gate independently sampled from an open set of dual-unitary gates~\citep{bertini2019exact}, i.e., gates which are unitary when viewed both bottom to top and left to right. The distribution of each gate has a non-zero weight on an open subset of the dual-unitary gates, which allows the spatial evolution to consist of a universal gate set as proved in Ref. \cite{ippoliti2023dynamical}. We focus on models without any global conservation laws, and on the regime  $t+1\geq N_A$.

We assume the system is measured in subsystem $B$ using a tensor product of two-site measurements. Precisely, the measurement  basis is generated from $\ket{\varphi_0}$ acted by the Heisenberg-Weyl operators~\citep{asadian2016heisenberg} $\sigma_\alpha$ as $\ket{\varphi_\alpha}=I\otimes\sigma_\alpha\ket{\varphi_0}$ with $\alpha=\{1,\cdots,d^2\}$. These measurement bases, generalizing the Bell measurement to the higher dimensions, have the maximum bipartite entanglement and are orthogonal to each other. Notably, they satisfy the solvable conditions for dual-unitary circuits~\citep{piroli2020exact}. 

We focus our discussion on the $k$-th moment $\rho_{N_A}^{(k)}$, diagrammatically shown in Fig.~\ref{fig1} (a) as a tensor network. The figure is depicted in the vectorized $k$-replica format, where forward and backward time evolutions are folded on top of each other. 
Each green gate consists of $k$ copies of forward and backward evolution
$
\begin{tikzpicture}[baseline=(current  bounding  box.center), scale=0.7]
\Wgategreen{0}{0}
\end{tikzpicture}
=(U\otimes U^{*})^{\otimes k}.
$
Moreover, the projective measurement with outcome $\alpha$ is denoted by an orange circle in Fig.~\ref{fig1} (b). We take the subsystems $B_1$ and $B_3$ to be infinitely long,  a condition which can be relaxed~\footnote{This assumption can be relaxed. As demonstrated in Refs.~\citep{suzuki2024global,riddell2025quantum}, the formation of unitary $k$-design is exponentially fast in the length of $B_1$ and $B_3$. 
Consequently, our results also hold for finite $B_1$ and $B_3$ as an $\epsilon$-approximate deep thermalization happens in a finite chain whose length scales as $\log(1/\epsilon)$.}
. 
In region $B_2$, the measurement outcomes on $m$ sites are assumed forgotten, which is equivalent to erasing these sites. Therefore, this region is traced out, i.e., the forward and backward time-sheets are connected.
In Fig.~\ref{fig1} (c), this partial trace is represented by $m$ empty bullets. Each empty bullet is proportional to the vectorized identity operator acting on the $k$ replicas, as
$
\begin{tikzpicture}[baseline=(current  bounding  box.center), scale=0.5]
\draw[very thick](0.4,0.4)--(0.8,0.8);
\MYcircle{0.8}{0.8}
\end{tikzpicture}
\!=\!
\frac{1}{d^{k/2}}
(\begin{tikzpicture}[baseline=(current  bounding  box.center), scale=0.6]
\draw[very thick](0,0)--(0.5,0.5)--(0.65,0.35)--(0.15,-0.15);
\end{tikzpicture})^{\otimes k}
=
\frac{1}{d^{k/2}}(\sum_{i=1}^d \ket{ii})^{\otimes k}
$ .

By an explicit calculation, we show deep thermalization for mixed states arising from such erasures of measurement outcomes or, equivalently, in the presence of totally depolarizing noise (see discussion and outlook).  
Concretely, for almost all realizations of the circuit model described above, we prove the following:\\

\emph{
1) Taking the circuit depth $t\to\infty$, our MSPE deeply thermalizes to the so-called generalized Hilbert-Schmidt ensemble, Eq.~\eqref{eq:HS_ensemble}.}

\emph{
2) At large $t$, this approach is exponentially quick.
For the  $k$-th moments to be $\epsilon$ close in both the two and trace norms, it takes time $t_k\sim\log(k/\epsilon)$.\\}

These two results hold for any ratio of $m/N_A$.
The generalized Hilbert-Schmidt ensemble is defined by tracing out $m$ qudits from a Haar random states on $N_A+m$ qudits:
\begin{equation}
    \mathcal{E}=\bigg\{\mathrm{Tr}_m \ket{\psi}_{N_A+m}\bra{\psi}\bigg{|} \ket{\psi} \sim \mathrm{Haar}(d^{N_A+m})\bigg\}.\label{eq:HS_ensemble}
\end{equation}
For $m=N_A$ this ensemble reduces to the so-called Hilbert-Schmidt ensemble known in the random matrix theory literature~\cite{sommers2004statistical,zyczkowski2011generating}.

The emergence of the generalized Hilbert-Schmidt ensemble as the limiting form of the MSPE under measurements which are lossy, constitutes one of the key results of our work. It is further a physically  reasonable result: 
intuitively, this ensemble can be viewed as the most ergodic mixed-state ensemble appearing under the scenario of lost measurement information. 
Indeed, suppose the number of erased measurement outcomes is zero ($m=0$), then the MSPE reduces to the  pure state ensemble of Haar random states expected from deep thermalization~\cite{cotler2023emergent,ho2022exact, Claeys_2022_emergent, Ippoliti_2022solvable, ippoliti2023dynamical},   which is the most  ergodic pure state distribution at infinite temperature (no conservation laws).
When $m > 0$, we  expect a loss of purity  for each state on $A$  that decreases exponentially with $m$; the generalized Hilbert-Schmidt ensemble indeed has this behavior.

We  briefly sketch the proof of our results (1) and (2). This involves a  computation of the moments $\rho_{N_A}^{(k)}$. For that, we analyze the evolution of the $k$-replicated system in the spatial direction from right to left. 
Assuming that $B_1$, $B_3$ are infinitely long, since the gate composed of the circuits form a universal gate set, the circuits there effectively form a unitary $k$-design~\cite{ho2022exact,ippoliti2023dynamical} 
in the space direction. 
The steady state, therefore, must commute with $U^{\otimes k}$ for any $U$ \cite{ippoliti2023dynamical}.
This subspace is spanned by the global permutations of $k$-replicas, so it suffices to track the coefficients of each permutation element. The erased sites in $B_2$ behave as totally depolarizing channels, which modify the coefficients of different permutation elements. Lastly, the gates in the subsystem $A$ implement an isometry that preserves these coefficients. Our results follow by combining all these effects and the replica trick, which we provide in more detail in the End Matter and the Supplemental Material~\citep{SM}.

Another interesting limit is to keep $t$ fixed ($t+1\geq N_A$) while increasing either $d$ or both $N_A,m$  to infinity. In the latter, the ratio $\gamma=m/N_A$ is fixed and $t+1\geq m$.
In the pure-state Haar random ensemble, the $k$-th moment is an equally weighted sum of permutation operators on the $k$-replicated Hilbert space. 
Here we find that instead the leading contribution (in $d$) for $\rho_{N_A}^{(k)}$ of MSPE comes only from the identity permutation, even if only a single measurement outcome is erased. The contributions from other permutation elements are suppressed as $1/d^m$ \footnote{However, for large but finite $d$, the highly degenerate subleading coefficients matter when considering the $k$-th moment for $k\gg d$,(see Eq.~(\ref{eq:nextOrderD}) in the End Matter).}. It is tempting to argue that in this limit, the MSPE is indistinguishable from the totally mixed state ensemble, where every conditional mixed state is identical to the maximally mixed state. However, this is not always the case. 
The difference from the totally mixed state ensemble can be resolved by, for example, looking at some nonlinear global quantities, such as Renyi-$k$ entropies, which can amplify the subleading contributions, or looking directly at the distribution of the eigenvalues of the mixed states~\footnote{We say more about distribution of the eigenvalues in the discussion and the Supplemental Material~\cite{SM}. }. 

\begin{figure}
\includegraphics[width=1.0\columnwidth]{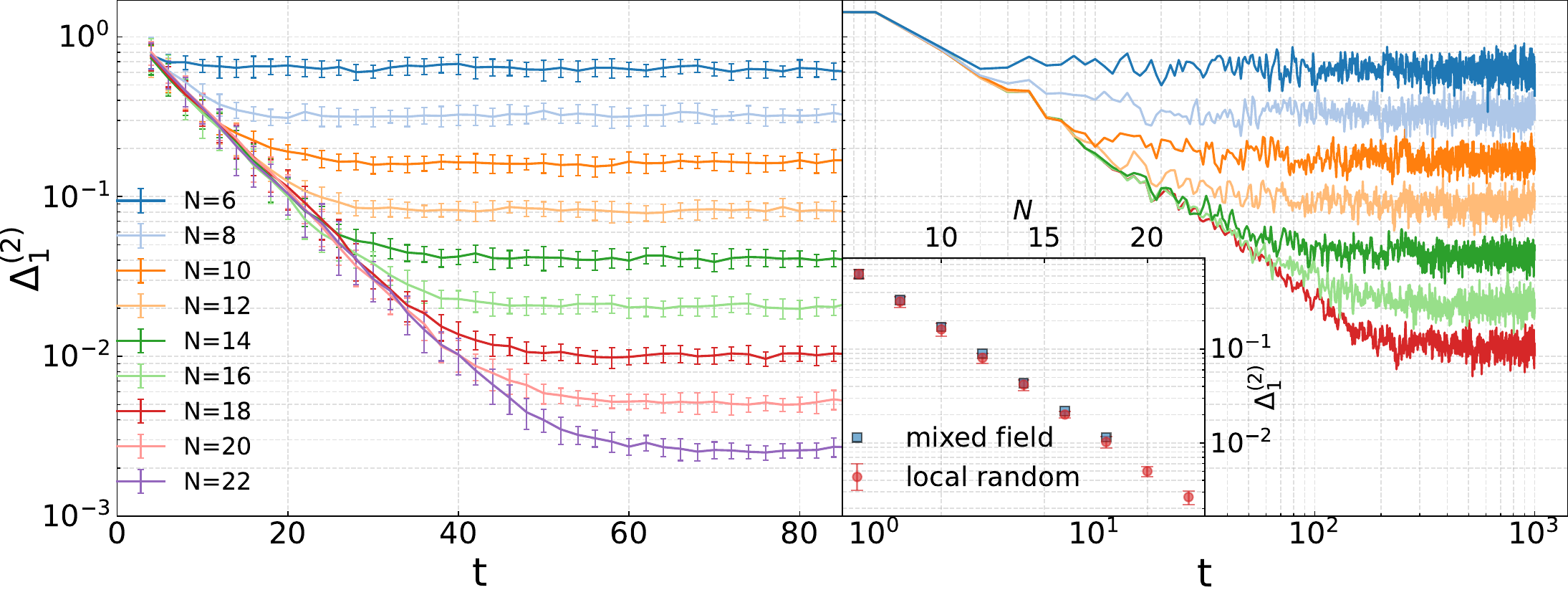}
\caption{
Distance of $\rho_{N_{A}}^{(k)}$ from the generalized Hilbert-Schmidt ensemble. Left: local Haar-random circuits for system size $N$ from $6$ to $22$, where error bars indicate statistical fluctuations over different circuit realizations. Right: mixed-field Ising dynamics for $N$ from $6$ to $18$, governed by the Hamiltonian $H=\sum_{j=1}^N(h_x\sigma_x^j+h_y\sigma^j_y)+\sum_{j=1}^{N-1}J\sigma_x^j\sigma_x^{j+1}$, with $\sigma_\alpha^j$ the Pauli matrix on $j$-th qubit and $(h_x,h_y,J)=(0.8090,0.9045,1.0)$.
Inset: saturated deviation $\Delta_1^{(2)}$ at late times versus the system size.
} 
\label{fig2}
\end{figure}

\emph{Generality of mixed state deep thermalization---}
Here we establish universality of mixed-state deep thermalization beyond the dual-unitary setting. 
Analytically, we can show that the limiting distribution of the MSPE is the generalized Hilbert–Schmidt ensemble   for global Haar-random circuits with $N_B$ larger than $\Omega(k N_A)$ and for local random unitary circuits with brick-wall structure in the large-$d$ limit,
see the End Matter and SM~\cite{SM}. 

We numerically demonstrate this by computing the distance $\Delta_1^{(2)}$ between the MSPE and the generalized Hilbert–Schmidt ensemble for both random unitary circuits and Hamiltonian mixed-field Ising dynamics with finite system size $N=N_A+N_B$ and local dimension $d=2$, see Fig.~\ref{fig2}. The results show that the saturation value of $\Delta_1^{(2)}$ in all cases decay  exponentially with system size, 
in line with our predictions. 
We further find that the decay rate depends on the dynamics (but not on the system size): for random unitary circuits, $\Delta_1^{(k)}$ decays exponentially in time before saturation, whereas for Hamiltonian dynamics the decay obeys a power law (this is presumably attributed to energy conservation).
In~\cite{SM} we provide further data on the kicked Ising model and dual-unitary circuit at finite system sizes. A detailed analysis of finite-time and finite-systems behavior across different models is left for future work.

\emph{Quantum conditional entropy and state teleportation---} We return  to our analytical setting of dual-unitary circuits. In the limit $d\to\infty$, the amount of information lost when even only a single measurement outcome is erased, $\log(d)$, diverges. This, as well as the fact that the leading (in $d$) contribution comes from the identity permutation, suggests that the MSPE should be indistinguishable from the totally mixed state. Nonetheless, the ratio between the lost and the retained information is always fixed to $m/N_A$. Therefore, how much quantum information can be transported through the region $B_2$ is not clear.
To address this, we examine the average ability of MSPE to {\it teleport} a quantum state through the bulk~\cite{bao2024finite,lovas2024quantum}. Concretely, we consider the following  set-up: we leave a single qudit at the rightmost boundary of the chain unmeasured, called the reference $R$, which will act as an input for an arbitrary state for our teleportation thought-experiment. We will try to recover the quantum information sent into $R$ from region $A$, mediated by the lossy measurements in $B$. Diagrammatically, this looks like:
\begin{equation*}
\begin{tikzpicture}[baseline=(current  bounding  box.center), scale=0.4]
\foreach \i in {0,2,5,7,9,12}
{
\foreach \j in {0,2}
{
\Wgategreen{\i}{\j}
}
}
\foreach \i in {-1,1,4,6,8,10,13}
{
\foreach \j in{1,3}
{
\Wgategreen{\i}{\j}
}
}
\Text[x=3.1,y=2]{$\bm{\cdots}$}

\Text[x=10.9,y=2]{$\bm{\cdots}$}

\foreach \i in {2,9,12}
{
\Measurement{\i}{4}
}
\foreach \i in{5,7}
{
\MYcircle{\i-0.5+0.1}{3.5+0.1}
}
\foreach \i in {6,8}
{
\MYcircle{\i-0.5-0.1}{3.5+0.1}
}
\foreach \i in {-2,0,3,5,7,9,12}
{
\draw[very thick]
  (\i+0.5,-0.5) to[out=-45, in=225] (\i+1.5,-0.5);
}
\foreach \i in {-1,1}
{
\draw[very thick]
  (13.5,\i+0.5) to[out=45, in=315] (13.5,\i+1.5);
}
\draw[very thick] (13.5,3.5)--(13.5,4);
\draw[very thick] (-0.5,3.5)--(-0.5,4);
\draw[very thick] (0.5,3.5)--(0.5,4);
\foreach \i in {-1,1}
{
\draw[very thick]
  (-1.5,\i+0.5) to[out=135, in=225] (-1.5,\i+1.5);
}
\draw[very thick] (-1.5,3.5) -- (-1.5,4);
\Text[x=-0.5,y=4.75]{\large{$\bm{A}$}}
\Text[x=13.5,y=4.75]{\large{$\bm{R}$}}
\draw[->,very thick](14.5,-0.5)--(14.5,5);
\Text[x=15,y=2.5] {\pmb{t}}
\end{tikzpicture}.
\end{equation*}
Now, the fidelity of teleportation can be upper bounded by the conditional entropy between the subsystem $A$ and the reference $R$~\citep{cerf1997negative}, which matches the coherent information~\citep{Schumacher1996Quantum} viewing evolution through the bulk as a quantum channel. 
For simplicity, we compute the annealed average Renyi-$k$ conditional entropy $I_{R:A}^{(k)}$
\begin{equation}
I_{R:A}^{(k)}=S_{AR}^{(k)}-S_{A}^{(k)}=-\frac{1}{k-1}\log\frac{\overline{\mathrm{Tr}\rho_{AR}^{k}}}{\overline{\mathrm{Tr}\rho_{A}^{k}}}, \label{eq:conditional_entropy_expression}
\end{equation}
where a bar indicates the average over the measurement outcomes  with associated Born probabilities.

The essence of computing the conditional entropy is to evaluate the Renyi-$k$ purities $\mathrm{Tr}(\overline{\rho^k})$,
expressible using the $k$-th moment of the density matrix ensemble. For example, $\mathrm{Tr}(\overline{\rho_A^2})=\mathrm{Tr}(\mathrm{SWAP}\rho_A^{\otimes 2})$ where $\mathrm{SWAP}$ is the standard swap operator over the two replicas and $\rho_A^{\otimes 2}=\overline{\rho_A\otimes \rho_A}$ is the second moment of $\rho_A$.  
Thus, the previously used methods including the replica trick apply here as well. 
We emphasize that different boundary conditions must be imposed on the reference $R$ site depending on whether $S_A^{(k)}$ or $S_{AR}^{(k)}$ is evaluated, which results in the non-trivial conditional entropy~\cite{SM}.

We provide the evaluation of Eq.~(\ref{eq:conditional_entropy_expression}) in the Supplemental Material~\cite{SM}.

In the limit $t\to\infty$ and $m,N_A\to\infty$ while the ratio $\gamma=m/N_A$ is fixed (or equivalently $d\to\infty)$, we can identify two phases for the capacity of quantum state teleportation. If $N_A>m$, $I_{R:A}^{(k)}=-\log(d)$ for any real $k\geq0$, which indicates the possibility for a perfect quantum state teleportation~\cite{cerf1997negative}. This agrees with our intuition at the beginning of this section. The extreme case $m=0$ reflects the fact that the space evolution is unitary, and no information is lost. For $m>0$, the erased sites are sparse compared to $N_A$, and the system still preserves enough quantum information to teleport the state. We call this the teleportation phase.

Conversely, if $N_A<m$, $I_{R:A}^{(k)}=\log(d)$, indicating that the subsystem $A$ and the reference $R$ are decoupled. This can be explained by the monogamy of entanglement~\cite{lovas2024quantum}. In this case, the number of erased sites $m$ is larger than the subsystem size $N_A$, and all the sites in $A$ are entangled with the environmental degrees of freedom~\footnote{We refer to the environmental degrees of freedom arising from the purification of tracing out $m$ sites in $B_2$.}; they cannot be entangled with the reference qudit $R$ anymore, and thus no state teleportation is possible. This corresponds to the decoupled phase.

The above two phases are separated by the transition point $N_A=m$, where the size of subsystem $A$ matches the number of erased sites. There $I^{(k)}_{R:A}=0$. Since this sharp transition holds for any $k\geq0$, we have analytically shown a sharp transition from the teleportation phase to the decoupled phase.

\emph{Discussion and outlook---}
In this Letter, we have introduced the notion of the mixed state projected ensemble (MSPE), a statistical description of a local subregion of a quantum many-body state, conditioned upon measurements on the complement which are imperfect. We have identified the limiting forms of the MSPE under the scenario of lossy measurements in dual-unitary quantum circuit evolution, results which represent a novel form of universality found in quantum equilibration. This can be understood as the conceptual generalization of deep thermalization, recently introduced for pure state ensembles, to the setting of mixed states.
Furthermore, we investigated the quantum information theoretic property of the MSPE arising in our model, 
and identified a phase transition in terms of the ability of the MSPE 
to
teleport quantum information between distant regions, which depends on the  ratio of the number of erased measurement outcomes to the size of the subsystem of interest. Our findings are particularly relevant for NISQ-era quantum devices, where imperfections and noise are unavoidable, and also provides a practical scheme to generate random states, which has applications in 
myriad quantum information tasks~\cite{Neill2018blueprint,Cross2019Validating,arute2019quantum,bouland2019complexity,Alagic2018Unforgeable}. An experimentally verifiable consequence of mixed-state deep thermalization is the emergence of the so-called Erlang distribution~\citep{Leemis01022008, Shaw2025Experimental, mandal2025partialprojectedensemblesspatiotemporal} governing the probability of probabilities  of bit-string outcomes~\citep{zhang2025holographicdeepthermalizationtheory,Shaw2025Experimental} upon measuring a projected state of the MSPE, different from the Porter–Thomas distribution known for pure-states~\citep{mark2024maximum}.

While in this work we have characterized the emergent random density matrix ensembles by their finite moments (c.f.~Eq.~\eqref{eq:moment_definition}), alternatively, one may study the induced eigenvalue-eigenvector distribution of density matrices on the subregion $A$. 
This is done in the Supplemental Material~\cite{SM}. 
For the MSPE and Hilbert-Schdmit ensembles studied in this work, the eigenvectors can be shown to always be distributed uniformly. The eigenvalue distribution for the MSPE in the limit $N_A\to\infty$ at fixed ratio $\gamma=m/N_A$  is more nontrivial:
we show the variance of the distribution is exponentially suppressed in $t$, so in this limit it reduces to a delta distribution matching the generalized Hilbert-Schmidt ensemble~\citep{sommers2004statistical,VA1967Distribution}. Remarkably, each instance in this limiting ensemble has a flat spectrum, which supports the intuition that the MSPE is an ergodic ensemble at infinite temperature for given $N_A$ and $m$.

Our work also opens new directions of inquiry. Here we have focused on the loss of measurement records on a small, non-extensive subregion ($m$ consecutive sites) in generating the MSPE. It would be interesting to understand the resulting universality in the case when the erased measurements have finite density. An immediately tractable 
setting is 
a sparse configuration where the erased measurements are spread randomly  in space, with sufficiently long regions of retained measurements in between. 
 At leading order of $d^{t}$, a quick computation yields the same results for deep thermalization and quantum state teleportation, see Supplemental Material~\citep{SM};  however, the two models appear to differ in their higher-order corrections of $d^{t}$. For dense measurement loss, we can prove that for global Haar-random circuits, the MSPE still deeply thermalizes to the generalized Hilbert-Schmidt ensemble~\citep{SM}. We conjecture it will be a general result for other dynamics. Notice that in the limit $m, N_A\to\infty$ while $N_A/m\to 0$, each instance in the generalized Hilbert-Schmidt ensemble is similar to the totally mixed state.

We may also consider different error models beyond the erasure errors  considered in this paper.
For example, discarding measurement outcomes 
is  equivalent to applying {\it totally} depolarizing channels to $B_2$ before measuring it.
This suggests a simple generalization ---application of 
{\it partial} depolarizing channel, 
e.g. defined by 
$\rho\to P\rho+(1-P)I/d$ (on $B_2$), with a tunable parameter $P$.
For $P=1$ we recover the pure state Haar ensemble, for $P=0$ the MSPE studied in this letter, and it remains an open question  what universal ensemble is approached in deep thermalization for general $P$. We may ask similar questions for other kinds of noisy quantum evolution and measurements, as well as dynamics in the presence of symmetries, conservation laws and beyond spin systems. 
For example, Ref.~\citep{milekhin2024observableprojectedensembles} considers related questions in conformal field theory with U(1) charges.

More fundamentally, it would be very interesting to understand how the universal mixed-state ensembles unveiled in this work, as well as those arising in future work discussed above, can be argued for in terms of general physical  principles like the principle of maximum entropy. Indeed, in deep thermalization of pure state ensembles, such a principle has been successfully formulated using ideas from quantum information theory, see Refs.~\cite{mark2024maximum, Liu2024GaussianDT,Varikuti_2024_unraveling,chang2025deepthermalizationchargeconservingquantum}, and can be understood as a generalization of the second law of thermodynamics. It would be very exciting to extend similar notions to mixed state projected ensembles, which would entail a significant  fundamental advancement of our understanding of quantum statistical mechanics.

\begin{acknowledgments} 
\emph{Acknowledgments.}---We thank  Ignacio Cirac for valuable discussions. W.~W.~H.~is supported by the National Research Foundation (NRF), Singapore, through the NRF Felllowship NRF-NRFF15-2023-0008, and through the National Quantum Office, hosted in A*STAR, under its Centre for Quantum Technologies Funding Initiative (S24Q2d0009).  P.~K.~acknowledges financial support from the Alexander
von Humboldt Foundation. 
\end{acknowledgments}

\bibliography{MyCollection}

\begin{appendices}
\appendix
\section*{End Matter}
\twocolumngrid
\section{$k$-th Moment}
\begin{figure}[t]
\centering
\begin{tikzpicture}[baseline=(current  bounding  box.center), scale=0.5]
\foreach \i in {0,2,5,9,12}
{
\foreach \j in {0,2,4}
{
\Wgategreen{\i}{\j}
}
}
\Wgategreen{7}{0}
\Wgategreen{7}{2}
\foreach \i in {-1,1,4,6,8,10,13}
{
\foreach \j in{1,3}
{
\Wgategreen{\i}{\j}
}
}
\Text[x=3.1,y=0]{$\bm{\cdots}$}
\Text[x=3.1,y=4]{$\bm{\cdots}$}
\Text[x=10.9,y=0]{$\bm{\cdots}$}
\Text[x=10.9,y=4]{$\bm{\cdots}$}
\foreach \i in {1,4,10,13}
{
\Measurement{\i}{5}
}
\foreach \i in{6}
{
\MYcircle{\i-0.5+0.1}{4.5+0.1}
}
\foreach \i in {9}
{
\MYcircle{\i-0.5-0.1}{4.5+0.1}
}
\MYcircle{6.6}{3.6}
\MYcircle{7.4}{3.6}
\foreach \i in {-2,0,3,5,7,9,12}
{
\draw[very thick]
  (\i+0.5,-0.5) to[out=-45, in=225] (\i+1.5,-0.5);
}
\foreach \i in {-1,1,3}
{
\draw[very thick]
  (13.5,\i+0.5) to[out=45, in=315] (13.5,\i+1.5);
}
\foreach \i in {-1,1,3}
{
\draw[very thick]
  (-1.5,\i+0.5) to[out=135, in=225] (-1.5,\i+1.5);
}
\draw[very thick,dashed](0.4,-1.5)--(0.4,6);
\draw [decorate, decoration={brace, amplitude=10pt}, very thick] (0.5,5.25) -- (4.5,5.25);
\draw [decorate, decoration={brace, amplitude=10pt}, very thick] (5.5,5.25) -- (8.5,5.25);
\draw [decorate, decoration={brace, amplitude=10pt}, very thick] (9.5,5.25) -- (13.5,5.25);
\Text[x=2.5, y=6.5]{\large{$\bm{B_3}$}}
\Text[x=7, y=6.5]{\large{$\bm{B_2}$}}
\Text[x=11.5, y=6.5]{\large{$\bm{B_1}$}}
\draw[->,very thick](14,-1.5)--(14,6);
\Text[x=14.5,y=2.5] {\pmb{t}}
\Text[x=0.4,y=-2]{\large$\bm{\rho_\mathrm{f}}$}
\draw[very thick,dashed,color=blue] (5.25,4.65)--(6.45,3.45)--(6.45,-1);
\Text[x=6.2,y=-1.5]{\large$\bm{\rho_\mathrm{a}}$}
\draw[very thick,dashed,color=blue] (8.65,4.65)--(7.45,3.45)--(7.45,-1);
\Text[x=7.7,y=-1.5]{\large$\bm{\rho_\mathrm{b}}$}
\draw[very thick, dashed,color=blue](13.45,5.0)--(13.45,-1);
\Text[x=13.45,y=-1.5]{\large$\bm{\rho_\mathrm{i}}$}
\draw[->,very thick](13.4,-2)--(4,-2);
\Text[x=8.7,y=-2.5]{Horizontal Evolution}
\end{tikzpicture}
\caption{Quench dynamics leading to MSPE. The symbols have the same meaning as in Fig.~\ref{fig1}. Here, we use the unitarity to simplify the 
 region $B_2$. We use the dashed line to indicate the states on the time-like cuts, which appear in the derivation of the results.}
\label{fig:End_matterkmomentSM}
\end{figure}

Here, we describe the calculation of the $k$-th moment $\rho_{N_A}^{(k)}$. 
It is convenient to consider the spatial evolution progressing from right to left by exchanging the roles of space and time. In this way, the initial state comprises of maximally entangled pairs over $t+1$ sites in all $k$ replica copies, as indicated in Fig.~\ref{fig:End_matterkmomentSM} $\rho_\mathrm{i}=\left(\ket{\varphi_0}^{\otimes (t+1)/2}\bra{\varphi_0}^{\otimes (t+1)/2}\right)^{\otimes k}$. 
The state evolved with $N_B$ layers of brick-wall circuit made of dual-unitary gates, which are also unitary for this evolution in space direction.
The original initial state becomes an open boundary condition on the bottom boundary (i.e. an identity matrix on a single site).
Projective measurements with outcome $\alpha$ instead imply the action of a single site unitary $\sigma_\alpha$ on the other (top) boundary. Therefore, if the measurement outcome is recorded, the evolution is unitary.

For any measurement outcome $\bm{\alpha}_1$ in $B_1$, we denote the spatial unitary in the subpart $B_1$ as $U_{\bm{\alpha}_1}$, whose unitarity is guaranteed by dual-unitary property of the circuit. Due to the assumption that the subpart $B_1$ is infinitely long, the  measurement outcome averaged state $\rho_\mathrm{b}$ on the time-like cut thus reaches its steady state before entering $B_2$, determined by the eigenvectors of the inhomogeneous spatial transfer matrix with the largest eigenvalues. Since any open set of the dual-unitary gate forms a universal gate set~\citep{suzuki2022computational}, $U_{\bm{\alpha}_1}$ then explores the whole space of the unitary group~\citep{ippoliti2023dynamical,riddell2025quantum} for different measurement outcomes. According to Ref.~\cite{ippoliti2023dynamical}, the steady state must commute with $U_{\bm{\alpha}_1}^{\otimes k}$ for any $\bm{\alpha}_1$, which is equivalent to any unitary. By Schur-Weyl duality, the subspace of the steady states 
is exclusively given by the representation $\rho_{t+1}(g)$ of a permutation element $g\in\mathbb{S}_k$ on the $k$ replicas, each with $t+1$ qudits living on the time-like cut.
Here $\mathbb{S}_k$ is the permutation group for $k$-replicas. Concretely, 
\begin{equation}
\begin{aligned}
&\rho_{t+1}(g)\\=&\sum_{i_1,i_2,\cdots,i_k=1}^{d^{t+1}}\ket{i_{g(1)}i_{g(2)}\cdots i_{g(k)}}\bra{i_1i_2\cdots i_k }.\label{eq:specific_representation}
\end{aligned}
\end{equation}
Note that this is a global state spanning all $t+1$ sites in each replica space, as indicated by the subscript of $\rho(g)$.
In the replica picture, the initial state is totally symmetric, thus $\rho_\mathrm{b}$ is an equal superposition of all these dominant eigenvectors
\begin{equation}
\rho_\mathrm{b} = \sum_g \rho_{t+1}(g).
\end{equation}
Here we omit a possible normalization factor, which does not influence the results.

In the region $B_2$, measurement outcomes are erased so that the action on the topmost $m/2$ qudits is given by a totally depolarizing quantum channel
$
   \frac{1}{d^2}\sum_\alpha \sigma_\alpha (\cdot) \sigma_\alpha^\dagger.
$
Using the unitarity of the circuit, region $B_2$ can be simplified as shown in Fig.~\ref{fig:End_matterkmomentSM}.
As a result, the time-like state after the depolarization channel, $\rho_\mathrm{a}$ in Fig.~\ref{fig:End_matterkmomentSM}, is no longer an eigenstate of the spatial transfer matrix.

Since $B_3$ is also infinitely long, it projects $\rho_\mathrm{a}$ back to the corresponding eigenspace. However, due to the replica structure in $\rho_\mathrm{a}$, the final steady state, $\rho_\mathrm{f}$ in Fig.~\ref{fig:End_matterkmomentSM}, does not correspond to the equal superposition of the permutations $\rho_{t+1}(g)$. 
The overlap between $\rho_{t+1}(g)$ and $\rho_\mathrm{f}$ can be determined from the overlap between $\rho_{t+1}(g)$ and $\rho_\mathrm{a}$ as $\rho_{t+1}(g)$ remains invariant under the evolution in $B_1$: $\mathrm{Tr}\left(\rho_{t+1}^\dagger(g)\rho_\mathrm{f}\right)=\mathrm{Tr}\left(\rho_{t+1}^\dagger(g)\rho_\mathrm{a}\right)$. 
It is important to note that the matrices $\rho_{t+1}(g)$ are not mutually orthogonal with respect to the Frobenius inner product. Taking this into account, we express $\rho_\mathrm{f}$ as a linear combination of $\rho_{t+1}(g)$: $\rho_\mathrm{f}=\sum_g \alpha(g)\rho_{t+1}(g)$ where $\alpha(g)$ satisfies a series of equations

\begin{equation}
\begin{aligned} & \sum_{g'}\frac{d^{ml(g)/2+ml(g')/2+(t+1-m/2)l(g^{-1}g')}}{d^{k(m/2+t+1)}}\\
= & \sum_{g'}\frac{d^{(t+1)l(g^{-1}g')}}{d^{(t+1)k}}\alpha(g')\ \mathrm{for}\ \forall g.
\end{aligned}\label{eq:equation_for_alpha_consecutive}
\end{equation}
Here $l(g)$ denotes the number of cycles~\citep{mehta1991random} in the permutation element $g$.


The state $\rho_\mathrm{f}$ completely determines the output $k$-th moment $\rho^{(k)}_{N_A}$ in region $A$.
Since the gates in the region $A$ are also dual-unitary,
they form an isometry $W$ from the time-like state to the spatial state with $WW^\dagger = I_{N_A}$.

This isometry does not mix representations of different permutation elements
as $W^{\otimes k}\rho_{t+1}(g)W^\dagger{}^{\otimes k}=\rho_{N_A}(g)$, see the Supplemental Material~\citep{SM} for details.
Here $\rho_{N_A}(g)$ is the representation of $g$ defined analogously to Eq. (\ref{eq:specific_representation}), except that each replica contains $N_A$ sites and thus the sum goes to $d^{N_A}$.
Therefore, $\rho_{N_A}^{(k)}$ has exactly the same weight of each permutation element as $\rho_\mathrm{f}$. We emphasize that the above calculation has not included the Born probability $P_{\bm{\alpha}}$, which can be included using the replica trick~\cite{SM}. Taking it into account  does not change the result. Consequently, the $k$-th moment of the distribution describing  MSPE is expressed as $\rho_{N_A}^{(k)}=\sum_g\alpha(g)\rho_{N_A}(g)$ and entirely determined by Eq. (\ref{eq:equation_for_alpha_consecutive}).

To the first order expansion of Eq.~(\ref{eq:equation_for_alpha_consecutive}) in large $t$ ($d^t$) limit, we can replace $l(g^{-1}g')$ with a delta function $d^{(t+1)l(g^{-1}g')}\approx d^{(t+1)k}\delta_{g,g'}$. Therefore, Eq. (\ref{eq:equation_for_alpha_consecutive}) simplifies to
\begin{equation}
\alpha(g)=\frac{d^{ml(g)}}{d^{mk}},
\label{eq:large_t_limit}
\end{equation}
which indicates that all the permutation elements contribute to the $k$-th moment. 
The coefficients $\alpha(g)$ in Eq. (\ref{eq:large_t_limit}) match the ones for the generalized Hilbert-Schmidt ensemble (cf. Eq.~\eqref{eq:HS_ensemble}). 
Thus, in the limit $t\to\infty$ the MSPE is deeply thermalized to it.

Let us now consider the next order in $1/d^t$, which will give us the deep thermalization time scales. 
Here, we need to expand the cycle function $l$ up to the next order. For that we have
\begin{equation}
\begin{aligned}
\frac{d^{(t+1)l(g^{-1}g')}}{d^{(t+1)k}}=&\delta_{g',g}+\\&\sum_{i<j}\frac{1}{d^{(t+1)}}\delta_{g',gs_{i,j}}+\mathcal{O}\bigg(\frac{1}{d^{2t}}\bigg),
\end{aligned}
\label{eq:EndMatter_nextleadinginGHS}
\end{equation}
where $s_{i,j}$ is the swap operator between the $i$-th and $j$-th replica.
Using this expansion, we evaluate the distance to the generalized Hilbert-Schmidt ensemble $\Delta_\xi^{(k)}$ with $\xi=1,2$. The full expression can be found in the Supplemental Material~\citep{SM}, which simplifies  for $k\ll d^m$ to
\begin{equation}
\begin{aligned}
\Delta^{(k)}_1&\lesssim\frac{1}{d^{t+1}}\frac{1}{d^{m}}(1-\frac{1}{d^{m}})\frac{k(k-1)}{2};\\
\Delta_{2}^{(k)}&\approx\frac{1}{d^{t+1}}\frac{1}{d^{m}}(1-\frac{1}{d^{m}})\sqrt{\frac{k(k-1)}{2}}+\mathcal{O}\bigg(\frac{1}{d^{2t}}\bigg).
\end{aligned}
\end{equation}
This indicates that the deep thermalization times scale as $t_k\sim\log(k/\epsilon)$ for an $\epsilon$-approximation to the generalized Hilbert-Schmidt ensemble for a given moment $k$. Note that higher moments require a longer time to achieve $\epsilon$-convergence 
leading to the separation of deep thermalization time scales~\citep{ippoliti2023dynamical}. 

On the other hand, if we take $d\to\infty$ first, the leading contribution in Eq. (\ref{eq:equation_for_alpha_consecutive}) is determined by $g=g'=e$ where $e$ is the identity element in the permutation group. This results in $\alpha(g)=\delta_{g,e}$, meaning that the $k$-th moment of the MSPE is always the $k$-fold tensor product of the identity 
density matrix
for any finite $k$. This corresponds to an ensemble where a totally mixed state is attributed to all measurement outcomes. The next order expansion in $d$ comes from the swap operator $s_{i,j}$ for arbitrary $i<j$. We have 
\begin{equation}
\alpha(s_{i,j})=1/d^m 
    \label{eq:nextOrderD}
\end{equation}
if $t>3m/2$, and the number of all the swap operators are $k(k-1)/2$, see Supplemental Material~\citep{SM}.

\section{Local Haar random circuits}
We consider the same circuit structure as in Fig.~\ref{fig1}, except that each local gate is now drawn Haar randomly.
In the limit $d\to\infty$, $N_{B_{1}},N_{B_{3}}\to\infty$,
and $t>\max{(N_{A},N_{B_{2}})}$, a typical realization yields an MSPE that is deeply thermalized to the generalized Hilbert–Schmidt ensemble.

The proof of this statement requires an extension of that in Ref.~\cite{Chan_2024_projected}. We show that after measuring subsystems $B_1$ and $B_3$ with outcome $\bm{\alpha}$, the projected ensemble on the remaining degrees of freedom, $\mathcal{E}=\{P_{\bm{\alpha}},\ket{\psi_{\bm{\alpha}}}_{A,B_{2}}\}$
is equivalent to the Haar-random state ensemble. Tracing out $B_2$ then yields the generalized Hilbert–Schmidt ensemble on subsystem $A$ following the definition Eq. (\ref{eq:HS_ensemble}). In the following, we briefly review the argument of Ref.~\cite{Chan_2024_projected} and extend it to show that the Haar-random ensemble emerges once $B_1$ and $B_3$ are measured. The difference is that~\cite{Chan_2024_projected} considered a bipartite system, whereas our system is partitioned into four parts.

In the limit $d\to\infty$, the circuit dynamics can be described using the membrane picture. Measurements impose distinct boundary conditions at the top of the circuit: one for regions $A,B_2$ and another for $B_1,B_3$.
Domain walls form at the interfaces between these regions and evolve backward in time, annihilating either at the boundary or with each other, as shown in Fig. \ref{Figure_domain_wall}. Each domain-wall trajectory incurs an energy penalty $d^{-kL}$, where $L$ is its length. For small $t$, the domain walls walk randomly in the
circuit. However, for $t>\max{(N_{A}, N_{B_2})}$, there exists a unique trajectory which
minimizes the energy penalty as shown in Fig. \ref{Figure_domain_wall}.
In the $d\to\infty$ limit, this trajectory dominates the frame potential, leading to the Haar-random state ensemble.

\begin{figure}
\includegraphics[width=0.6\columnwidth]{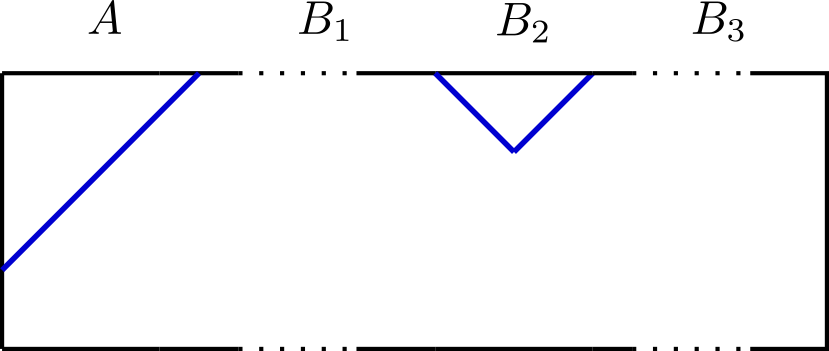}

\caption{Membrane-picture illustration of domain-wall evolution (blue lines). The figure shows the trajectory with minimal energy penalty.}

\label{Figure_domain_wall}
\end{figure}
\end{appendices}
\newpage
\onecolumngrid

\newcounter{equationSM}
\newcounter{figureSM}
\newcounter{tableSM}
\stepcounter{equationSM}
\setcounter{equation}{0}
\setcounter{figure}{0}
\setcounter{table}{0}
\makeatletter
\setcounter{equation}{0}
\setcounter{figure}{0}
\setcounter{table}{0}
\setcounter{section}{0}
\makeatletter
\renewcommand{\theequation}{S\arabic{equation}}
\renewcommand{\thefigure}{S\arabic{figure}}
\renewcommand{\thetable}{S\arabic{table}}
\begin{center}
{\large{\bf Supplemental Material for\\
 ``Mixed state deep thermalization''}}
\end{center}

\begin{itemize}
\item[-] In Section~\ref{Sec:universalgateset} we prove that our choice of dual-unitary gates form a universal gate set;
\item[-] In Section~\ref{Sec:k_th_moment_calculation} we derive the general equations determining the non-normalized $k$-th moment of MSPE;
\item[-] In Section~\ref{Sec:large_t} we discuss the $k$-th moment of MSPE in the limit $t\to\infty$; 
\item[-] In Section~\ref{Sec:correction_of_large_t} we discuss the deep thermalization of MSPE to the Hilbert Schmidt ensemble and the corrections at finite $t$;
\item[-] In Section~\ref{Sec:large_d} we discuss the $k$-th moment of MSPE in the limit $d\to\infty$;
\item[-] In Section~\ref{Sec:Eigenvalue_distribution} we compute and discuss the eigenvalue distribution of MSPE;
\item[-] In Section~\ref{Sec:DT_typical} we provide more analytical and numerical results for other typical models which deeply thermalize;
\item[-] In Section~\ref{Sec:ConditionalInformation} we compute the conditional entropy. The setup is described in the main text;
\item[-] In Section~\ref{Sec:Modelwithsparseerausureerrors}, we compute the MSPE with sparse $m$ erasure errors on the measurement outcomes.
\end{itemize}

\section{Universal gate set\label{Sec:universalgateset}}
In this section, we show that any open set $\mathcal{G}$ of dual-unitary
two-qudit gates arranged in a brick-wall circuit is so called brick-wall universal, see the following definition
and App. E of Ref. \cite{ippoliti2023dynamical}. As a consequence, averaging
spatial evolution $\{U_{\bm{\alpha}}^{\otimes k}\}$ over $\bm{\alpha}$
produces a unitary $k$-design, which directly follows the argument from Refs.
\cite{ho2022exact,ippoliti2023dynamical}. We use this result in the End
Matter to determine the spatial steady state.

Our proof follows a similar logic to the proof in App. E of Ref. \cite{ippoliti2023dynamical}. 

\emph{Definition}.---A set $\mathcal{G}\subset U(d^{2})$ of two-qudit
gate is ``brick-wall universal'' if the set of brick-wall circuits
on $t+1$ qudits, with open boundary conditions, composed of gates
in $\mathcal{G}$, is dense in the unitary group $U(d^{t+1})$.

Alternatively, a set $\mathcal{G}$ is brick-wall universal if it
can generate a universal gate set out of the brick-wall structure
on the $t+1$ qudits. Notice that even if $\mathcal{G}$ is a universal
gate set, it may not be brick-wall universal, as here we impose restrictions on the architecture of the circuits. 

In the following, we will show that any open set of the dual-unitary
gate $\mathcal{G}\subset U(d^{2})$ is brick-wall universal. Specifically,
we will show that this set $\mathcal{G}$ can generate arbitrary single-site
rotation on any site, together with two-qudit imprimitive gates on
any nearest-neighbor sites. According to Ref. \cite{brylinski2001universalquantumgates},
a gate is primitive if it is equivalent to the Identity or SWAP gate
up to single-site rotation, and only the latter is dual-unitary. An
imprimitive gate plus an arbitrary single-site rotation forms a universal
gate set.

\begin{proof}

Let us take any gate $U\in\mathcal{G}$. From $U$, we can construct
a brick-wall layer in the spatial direction as $\mathbb{U}=(I_{d}\otimes U^{\otimes(t-1)/2}\otimes I_{d})\cdot U^{\otimes(t+1)/2}$.
This is the unitary matrix with eigenphases $\{\phi_{i},i=1,\cdots d^{t+1}\}$.
The distance of $\mathbb{U}^{j}$ from the identity is given by $\lVert\mathbb{U}^{j}-I_{d}^{\otimes(t+1)}\rVert_{2}^{2}=2\sum_{i}[1-\cos(\phi_{i}j)]$,
a quasiperiodic function. This distance can be made arbitrarily small
by choosing a suitable $j=J\in\mathbb{N}$, according to the Poincaré
recurrence theory, thus $\mathbb{U}^{J}=I_{d}^{\otimes(t+1)}$. Here,
we directly write the equal sign because we can choose the distance
to be smaller than any prescribed precision. More rigorously, one
should keep an error $\epsilon$ and let $\epsilon\to0$ at the final
step. Hence, we have $\mathbb{U}^{J-1}=\mathbb{U}^{\dagger}$, i.e.,
we can also realize $\mathbb{U}^{\dagger}$ from the brick-wall circuit
built from $\mathcal{G}$.

It is known that the dual-unitary gate has the freedom of a single-site
rotations on four legs. Since $\mathcal{G}$ is an open set, the gate
$V=(v_{l_{1}}\otimes v_{l_{2}})\cdot U\cdot(v_{r_{1}}\otimes v_{r_{2}})$
also belongs to $\mathcal{G}$ for small enough single-site rotation
$v_{l_{1}},v_{l_{2}},v_{r_{1}},v_{r_{2}}$. Now, we construct $\mathbb{V}$
as a brick-wall layer from $\mathbb{U}$ by replacing one of $U$
at the first layer with $V$, which acts on the sites $(i,i+1)$ for
some $i\in\mathrm{even}$. The brick-wall circuit $\mathbb{U}^{J-1}\mathbb{V}$
gives us a two-qudit gate 

\begin{equation}
\mathbb{U}^{J-1}\mathbb{V}=\mathbb{U}^{\dagger}\mathbb{V}=U^{\dagger}V=U^{\dagger}(v_{l_{1}}\otimes v_{l_{2}})U(v_{r_{1}}\otimes v_{r_{2}})\label{SM:eq:the_gate_we_obtain_universal_sec}
\end{equation}
on the sites $(i,i+1)$ and Identity elsewhere. By choosing $v_{l_{1}}=v_{l_{2}}=I_{d}$,
we realize $v_{r_{1}}$, $v_{r_{2}}$ on the sites $i$, $i+1$, respectively.
Therefore, we can realize any single-site rotation out of the brick wall
circuit built from $\mathcal{G}$.

We still need to prove that $\mathcal{G}$ can produce the imprimitive
gate on any nearest-neighbor sites. If $U$ is already imprimitive,
we can choose $v_{r_{1}}=v_{r_{2}}=I_{d}$ in Eq. (\ref{SM:eq:the_gate_we_obtain_universal_sec}).
Then the circuit $\mathbb{U}^{J-1}\mathbb{V}=U^{\dagger}(v_{l_{1}}\otimes v_{l_{2}})U$
 produces an imprimitive gate on sites $i,i+1$ and Identity elsewhere
for some $v_{l_{1}},v_{l_{2}}$.

On the other hand, if $U$ is a primitive dual unitary gate, it can
be written as $U=(u_{l_{1}}\otimes u_{l_{2}})\cdot\mathrm{SWAP}$.
However, $V=(u_{l_{1}}\otimes u_{l_{2}})\cdot\mathrm{SWAP}\cdot D$
is also a dual-unitary with $D$ a diagonal two-qudit gate. Here, we
choose $V$ and $U$ to have the same single-site rotations. As long
as we choose $D$ to be very close to the Identity (but still imprimitive),
the gate $V$ will also be inside the open set $\mathcal{G}$. From
Eq. (\ref{SM:eq:the_gate_we_obtain_universal_sec}), we can generate
an imprimitive gate $D$ from $\mathbb{U}^{J-1}\mathbb{V}$ on the
sites $(i,i+1$). Therefore, we can always produce imprimitive gates
on any sites $(i,i+1)$ with $i\in\mathrm{even}$. The same argument
holds for the sites $(i,i+1)$ with $i\in\mathrm{odd}$, if we replace
one of the gates $U$ in the second layer of $\mathbb{U}$ in the
construction of $\mathbb{V}$ and consider the circuit $\mathbb{V}\mathbb{U}^{J-1}$.

\end{proof}

\section{Calculation of k-moment}
\label{Sec:k_th_moment_calculation}

In this section, we will calculate the $k$-th moment of MSPE on the subsystem $A$: $\rho_A^{(k)}$. For convenience, we denote the matrix inner product as 
\begin{equation}
\langle A,B \rangle=\frac{\mathrm{Tr}\left(A^\dagger B\right)}{D},
\end{equation}
with $D$ being the underlying Hilbert space dimension.

We first define the non-normalized mixed state 
$\tilde{\rho}_{N_{A}}(\bm{\alpha})$ as 
\begin{equation}
\tilde{\rho}_{N_{A}}(\bm{\alpha})=\mathrm{Tr}_{B_{2}}[(I_{A+B_{2}}\otimes\bra{\bm{\alpha}})\ket{\Psi}\bra{\Psi}(I_{A+B_{2}}\otimes\ket{\bm{\alpha}})],
\end{equation}
and the reduced moment $\tilde{\rho}_{N_{A}}^{(k)}$ as 
\begin{equation}
\tilde{\rho}_{N_{A}}^{(k)}=C_{k}\sum_{\bm{\alpha}}\tilde{\rho}_{N_{A}}^{\otimes k}(\bm{\alpha}).
\end{equation}
Here $C_{k}$ is a normalization factor which can be chosen as $\mathrm{Tr}\tilde{\rho}_{N_{A}}^{(k)}=1$ for convenience.

For simplicity, we introduce the number of lost measurement outcomes $n=m/2$, which is half of the number of sites in $B_2$.  Due to the space-time duality, we can consider the spatial evolution of states living on the timelike cut, progressing from right to left, as denoted by $\rho_{\mathrm{f}},\rho_{\mathrm{a}},\rho_\mathrm{b},\rho_\mathrm{i}$ in Fig.~\ref{fig:kmomentSM}. Those are states on $t+1$ qudits. 
Rigorously speaking, the output of Fig.~\ref{fig:kmomentSM} for a specfic measurement outcome $\bm{\alpha}$ is $\tilde{\rho}_{N_A}(\bm{\alpha})$, and the sum over all $\bm{\alpha}$ is proportional to $\tilde{\rho}_{N_A}^{(k)}$. In the following, without explicitly speaking, the figure implies a sum over $\bm{\alpha}$. 
Due to the assumption that $B_1$ is infinitely long, and the fact that the initial state $\rho_\mathrm{i}$ is a tensor product one among different replicas, the state $\rho_\mathrm{b}$ is a totally symmetric state, described by the equal superposition of the representation of permutation elements $\{g\}$.
\begin{equation}
\rho_\mathrm{b}=\sum_g \rho_{t+1,k}(g).
\end{equation}
Unless explicitly stated, we discard a possible normalization factor of the state $\rho$ in this section. Here $\rho_{t+1,k}(g)$ is the representation of the permutation element $g$ across $k$ replicas, as defined in the main text. The subscript $k$ of $\rho(g)$ indicates the number of replicas the operator is supported on, and the subscript $t+1$ indicates the number of sites in each replica. Concretely,
\begin{equation}
\rho_{t+1,k}(g)=\sum_{i_1,i_2,\cdots,i_k=1}^{d^{t+1}}\ket{i_{g(1)}i_{g(2)}\cdots i_{g(k)}}\bra{i_1i_2\cdots i_k },\label{eq:SM_specific_representation}
\end{equation}
where $d^{t+1}$ is the Hilbert space dimension of each replica with $t+1$ sites.

\begin{figure}[t]
\centering
\begin{tikzpicture}[baseline=(current  bounding  box.center), scale=0.5]
\foreach \i in {0,2,5,9,12}
{
\foreach \j in {0,2,4}
{
\Wgategreen{\i}{\j}
}
}
\Wgategreen{7}{0}
\Wgategreen{7}{2}
\foreach \i in {-1,1,4,6,8,10,13}
{
\foreach \j in{1,3}
{
\Wgategreen{\i}{\j}
}
}
\Text[x=3.1,y=0]{$\bm{\cdots}$}
\Text[x=3.1,y=4]{$\bm{\cdots}$}
\Text[x=10.9,y=0]{$\bm{\cdots}$}
\Text[x=10.9,y=4]{$\bm{\cdots}$}
\foreach \i in {1,4,10,13}
{
\Measurement{\i}{5}
}
\foreach \i in{6}
{
\MYcircle{\i-0.5+0.1}{4.5+0.1}
}
\foreach \i in {9}
{
\MYcircle{\i-0.5-0.1}{4.5+0.1}
}
\MYcircle{6.6}{3.6}
\MYcircle{7.4}{3.6}
\foreach \i in {-2,0,3,5,7,9,12}
{
\draw[very thick]
  (\i+0.5,-0.5) to[out=-45, in=225] (\i+1.5,-0.5);
}
\foreach \i in {-1,1,3}
{
\draw[very thick]
  (13.5,\i+0.5) to[out=45, in=315] (13.5,\i+1.5);
}
\foreach \i in {-1,1,3}
{
\draw[very thick]
  (-1.5,\i+0.5) to[out=135, in=225] (-1.5,\i+1.5);
}
\draw[very thick,dashed](0.4,-1.5)--(0.4,6);
\draw [decorate, decoration={brace, amplitude=10pt}, very thick] (0.5,5.25) -- (4.5,5.25);
\draw [decorate, decoration={brace, amplitude=10pt}, very thick] (5.5,5.25) -- (8.5,5.25);
\draw [decorate, decoration={brace, amplitude=10pt}, very thick] (9.5,5.25) -- (13.5,5.25);
\Text[x=-1,y=5.5]{\large{$\bm{\tilde{\rho}_{N_A}^{(k)}}$}}
\Text[x=2.5, y=6.5]{\large{$\bm{B_3}$}}
\Text[x=7, y=6.5]{\large{$\bm{B_2}$}}
\Text[x=11.5, y=6.5]{\large{$\bm{B_1}$}}
\draw[->,very thick](14,-1.5)--(14,6);
\Text[x=14.5,y=2.5] {\pmb{t}}
\Text[x=0.4,y=-2]{\large$\bm{\rho_\mathrm{f}}$}
\draw[very thick,dashed,color=blue] (5.25,4.65)--(6.45,3.45)--(6.45,-1);
\Text[x=6.2,y=-1.5]{\large$\bm{\rho_\mathrm{a}}$}
\draw[very thick,dashed,color=blue] (8.65,4.65)--(7.45,3.45)--(7.45,-1);
\Text[x=7.7,y=-1.5]{\large$\bm{\rho_\mathrm{b}}$}
\draw[very thick, dashed,color=blue](13.45,5.0)--(13.45,-1);
\Text[x=13.45,y=-1.5]{\large$\bm{\rho_\mathrm{i}}$}
\draw[->,very thick](13.4,-2)--(4,-2);
\Text[x=8.7,y=-2.5]{Horizontal Evolution}
\end{tikzpicture}
\caption{Quench dynamics leading to MSPE. The symbols have the same meaning as in Fig.~\ref{fig1}
in the main text. Here, we use the unitarity to simplify the $B_2$ subsystem. We also use the dashed line to represent the states on the timelike cut.}
\label{fig:kmomentSM}
\end{figure}

In the subsystem $B_2$, the state undergoes totally depolarization channels on its topmost $n$ qudits. Thus, the state leaving $B_2$, i.e., $\rho_\mathrm{a}$ in Fig.~\ref{fig:kmomentSM}, can be written as

\begin{equation}
    \rho_\mathrm{a}=\frac{\sum_g \mathrm{Tr}_n(\rho_{t+1,k}(g))\otimes I_n}{d^{nk}}.
\end{equation}
Here $\mathrm{Tr}_n$ and $I_n$ are the partial trace and Identity on the topmost $n$ qudits, respectively. In the following, we introduce the function $l(g)$ counting the number of cycles in the permutation element $g$ and using the letter $e$ denoting the Identity permutation element. The above expression can be rewritten as 
\begin{equation}
    \rho_{\mathrm{a}}=\sum_g d^{n\left(l(g)-k\right)}\rho_{t+1-n,k}(g)\otimes \rho_{n,k}(e).\label{eq:SM_rho_a}
\end{equation}

Due to the assumption that the subsystem $B_3$ is also  infinitely long, the state $\rho_{\mathrm{a}}$ is further projected into the subspace with eigenvalues $1$
of the spatial transfer matrix, which is spanned by all representations of the permutation elements $\{\rho_{t+1,k}(g),g\in\mathbb{S}_k\}$. Thus, 
\begin{equation}
\rho_\mathrm{f}=\sum_g \alpha(g) \rho_{t+1,k}(g), \label{Eq.definition_alphag}
\end{equation}
with the coefficients $\alpha(g)$ to be determined. 
Since $\rho_{t+1,k}(g)$ is invariant under the spatial transfer matrix, its inner products with $\rho_\mathrm{f}$ and $\rho_\mathrm{a}$ are the same, i.e., $\langle \rho_{t+1,k}(g),\rho_\mathrm{f} \rangle=\langle \rho_{t+1,k}(g),\rho_\mathrm{a}\rangle:=c(g)$. We can first calculate $c(g)$ from Eq.~\eqref{eq:SM_rho_a} as
\begin{equation}
\begin{aligned}
    c(g)&=\sum_{g'} d^{n\left(l(g')-k\right)}\left\langle \rho_{t+1,k}^{\dagger}(g),\rho_{t+1-n,k}(g')\otimes \rho_{n,k}(e)\right\rangle\\
    &=\sum_{g'}d^{n\left(l(g')-k\right)}d^{n\left(l(g)-k\right)}d^{(t+1-n)\left(l(g^{-1}g')-k\right)}. \label{eq:result_of_cg}
\end{aligned}
\end{equation}
On the other hand, we can obtain the relation between $\alpha(g)$ and $c(g)$ from Eq. (\ref{Eq.definition_alphag}) as 
\begin{equation}
c(g)=\sum_{g'}\alpha(g') \frac{d^{(t+1)\times l(g^{-1}g')}}{d^{(t+1)k}}.\label{eq:relation_between_cg_alphag}
\end{equation}
Equations (\ref{eq:result_of_cg}) and (\ref{eq:relation_between_cg_alphag}) provide a series of linear equations of $\alpha(g)$. In principle, we can solve these linear equations to obtain the desired $\alpha(g)$.

The dual-unitry gates in the subsystem $A$ form an isometry mapping from $t+1$ sites along the vertical cut to $N_A$ sites along the horizontal output. When $t+1\geq N_A$, the $k$-th reduced moment on the subsystem $A$ $\tilde\rho_{N_A}^{(k)}$ is expressed as
\begin{equation}
    \tilde\rho_{N_A}^{(k)}\propto W^{\otimes k}\rho_{\mathrm{f}} W^{\dagger}{}^{\otimes k},
\end{equation}
with $WW^\dagger=I_{N_A}$.
We can further identify that 
\begin{equation}
\begin{aligned}  \bra{j_{1}\cdots j_{k}}W^{\otimes k}\rho_{t+1,k}(g)W^{\dagger}{}^{\otimes k}\ket{j'_{1}\cdots j'_{k}}
= & \sum_{i_{1}\cdots i_{k}}\prod_{u=1}^{k}W_{j_{u}i_{g(u)}}\prod_{u'=1}^{k}W_{j'_{u'}i_{u'}}^{*}
=  \sum_{i_{1}\cdots i_{k}}\prod_{u=1}^{k}W_{j_{g^{-1}(u)}i_{u}}\prod_{u'=1}^{k}W_{j'_{u'}i_{u'}}^{*}\\
= & \prod_{u=1}^{k}\sum_{i_{u}}W_{j_{g^{-1}(u)}i_{u}}W_{j'_{u}i_{u}}^{*}
= \prod_{u=1}^{k}\delta_{j_{g^{-1}(u)},j'_{u}}
=  \delta_{j_{1},j'_{g(1)}}\delta_{j_{2},j'_{g(2)}}\cdots\delta_{j_{k},j'_{g(k)}}.
\end{aligned}
\end{equation}
Here, in the second equality, we relabel $u\to g^{-1}(u)$. In the third equality, we reorder the product and the sum. In the last equality, we relabel again $u\to g(u)$. Consequently, we obtain that $W^{\otimes k}\rho_{t+1,k}(g)W^\dagger{}^{\otimes k}=\rho_{N_A,k}(g)$. 

Therefore, $\rho_{N_A}^{(k)}$ is still in the permutation space with the same coefficient $\alpha(g)$ as
\begin{equation}
\tilde\rho_{N_A}^{(k)}\propto\sum_g \alpha(g) \rho_{N_A,k}(g).
\end{equation}

\section{Leading order in $t\to\infty$ limit}
\label{Sec:large_t}
\subsection{Expression for $\tilde{\rho}_{N_A}^{(k)}$}
Here we can consider finite $d$ with $t\to\infty$. In this limit, we only keep the leading order of $d^t$ as
\begin{equation}
\begin{aligned}c(g) & =\frac{\sum_{g'}d^{nl(g')+nl(g)+(t+1-n)\times l(g^{-1}g')}\delta_{g',g}}{d^{k(t+1+n)}}+\mathcal{O}\left(\frac{1}{d^{t}}\right)\\
 & =\frac{d^{2nl(g)}}{d^{2nk}}+\mathcal{O}\left(\frac{1}{d^{t}}\right)=\frac{1}{d^{2n(k-l(g))}}+\mathcal{O}\left(\frac{1}{d^{t}}\right),
\end{aligned}
\end{equation}
and
\begin{equation}
c(g)=\sum_{g'}\alpha(g')\frac{d^{(t+1)\times l(g^{-1}g')}}{d^{(t+1)k}}=\alpha(g)+\mathcal{O}\left(\frac{1}{d^{t}}\right).
\end{equation}
Thus, we obtain that under the limit $t\to\infty$, 
\begin{equation}
\alpha(g)=\frac{1}{d^{2n(k-l(g))}}+\mathcal{O}\left(\frac{1}{d^{t}}\right). \label{eq:SM_MSPE_Leading_order_large_t}
\end{equation}

After taking the normalization factor $C_k$ into account, we have that in the limit $t\to\infty$,
\begin{equation}
\tilde{\rho}_{N_{A}}^{k}=\frac{\sum_{g\in\mathbb{S}_{k}}\rho_{N_{A},k}(g)d^{2nl(g)}}{(d^{N_{A}+2n}+k-1)!/(d^{N_{A}+2n}-1)!}.
\end{equation}

\subsection{Expression of $\rho_{N_A}^{(k)}$}
In order to go from $\tilde{\rho}_{N_A}^{(k)}$ to $\rho_{N_A}^{(k)}$, we have to recover the Born probability, which can be achieved by the replica trick.
We first rewrite $\rho_{N_{A}}^{(k)}$ as 
\begin{equation}
\rho_{N_{A}}^{(k)}=\sum_{\bm{\alpha}}\tilde{\rho}_{N_{A}}^{\otimes k}(\bm{\alpha})[\mathrm{Tr}\tilde{\rho}_{N_{A}}(\bm{\alpha})]^{1-k}.
\end{equation}
In the replica trick, we will compute  
\begin{equation}
\bar{\rho}_{N_{A}}^{(k,q)}=C_{k+q}\sum_{\bm{\alpha}}\tilde{\rho}_{N_{A}}^{\otimes k}(\bm{\alpha})[\mathrm{Tr}\tilde{\rho}_{N_{A}}(\bm{\alpha})]^{q}
\end{equation}
for every integer $q$, and then take the limit $q=1-k$. Here we
choose $C_{k+q}$ as a global normalization factor such that $\mathrm{Tr}\overline{\rho}_{N_{A}}^{(k,q)}=1$
for convenience. Notice that $C_{1}=1$. As we will see later, this limit exists because the result will be the same for every $q$.

$\overline{\rho}_{N_{A}}^{(k,q)}$ can be related to $\tilde{\rho}_{N_{A}}^{(k+q)}$
for integer $q>0$ as 
\begin{equation}
\overline{\rho}_{N_{A}}^{(k,q)}=\mathrm{Tr}_{\{k+q\}/\{k\}}\tilde{\rho}_{N_{A}}^{(k+q)}, \label{SM:expression_of_bar_rho_k_q}
\end{equation}
where $\mathrm{Tr}_{\{k+q\}/\{k\}}$ is the trace over the last $q$
replicas, which corresponds to having the Identity permutation as the
boundary of those replicas. According to the previous subsection, at
late time $t\to\infty$, we have
\begin{equation}
\tilde{\rho}_{N_{A}}^{(k+q)}=\frac{\sum_{g\in\mathbb{S}_{k+q}}\rho_{N_{A},k+q}(g)d^{2nl(g)}}{(d^{N_{A}+2n}+k+q-1)!/(d^{N_{A}+2n}-1)!}.
\end{equation}
To illustrate the computation of $\mathrm{Tr}_{\{k+q\}/\{k\}}$, we
can start from $q=1$. It is well known that any permutation of $k+1$
elements has a unique decomposition into a permutation on the first
$k$ elements times a SWAP on the $(k+1)$-th element, expressed as
\begin{equation}
g=g's_{j,k+1}.
\end{equation}
Here $g\in\mathbb{S}_{k+1}$, $g'\in\mathbb{S}_{k}$, $s_{j,k+1}$
is the swap between the $j$-th and $(k+1)$-th element, with $s_{k+1,k+1}$
the identity permutation. By a direct computation, we obtain 
\begin{equation}
\mathrm{Tr}_{\{k+1\}/\{k\}}\rho_{N_{A},k+1}(g)=\begin{cases}
\rho_{N_{A},k}(g'), & \mathrm{if}\ j<k+1;\\
d^{N_{A}}\rho_{N_{A},k}(g'), & \mathrm{if}\ j=k+1,
\end{cases}\label{eq:SM_trace_replica_Eq1}
\end{equation}
and
\begin{equation}
l(g)=\begin{cases}
l(g') & \mathrm{if}\ j<k+1;\\
l(g')+1 & \mathrm{if}\ j=k+1.
\end{cases}\label{eq:SM_trace_replica_Eq2}
\end{equation}
These relations immediately lead to 
\begin{equation}
\overline{\rho}_{N_{A}}^{(k,1)}=\frac{\sum_{g'\in\mathbb{S}_{k}}\rho_{N_{A},k}(g')(k+d^{N_{A}+2n)})d^{2nl(g')}}{(d^{N_{A}+2n}+k)!/(d^{N_{A}+2n}-1)!}=\frac{\sum_{g'\in\mathbb{S}_{k}}\rho_{N_{A},k}(g')d^{2nl(g')}}{(d^{N_{A}+2n}+k-1)!/(d^{N_{A}+2n}-1)!}.
\end{equation}

We can systematically iterate the above trace computation over the
replica for every integer $q$. Specifically, we start from the state
$\tilde{\rho}_{N_{A}}^{(k+q)}$ defined on $(k+q)$ replicas. By using
Eqs. (\ref{eq:SM_trace_replica_Eq1}) and (\ref{eq:SM_trace_replica_Eq2})
to trace over the $(k+q)$-th replica, we obtain a state defined on
$(k+q-1)$ replicas. We then apply the same procedure to trace over
the $(k+q-1)$-th replicas, yielding a state on $(k+q-2)$ replicas.
Repeating this process iteratively, we eventually arrive at the desired
state defined on $k$ replicas, which reads as
\begin{equation}
\overline{\rho}_{N_{A}}^{(k,q)}=\frac{\sum_{g'\in\mathbb{S}_{k}}\rho_{N_{A},k}(g')\prod_{j=0}^{q-1}(k+j+d^{N_{A}+2n})d^{2nl(g')}}{(d^{N_{A}+2n}+k+q-1)!/(d^{N_{A}+2n}-1)!}=\frac{\sum_{g'\in\mathbb{S}_{k}}\rho_{N_{A},k}(g')d^{2nl(g')}}{(d^{N_{A}+2n}+k-1)!/(d^{N_{A}+2n}-1)!}.
\end{equation}
This result is independent of $q$, and therefore, we make the analytical continuation to any $q=1-k$, which gives us 
\begin{equation}
\rho_{N_{A}}^{(k)}=\frac{\sum_{g'\in\mathbb{S}_{k}}\rho_{N_{A},k}(g')d^{2nl(g')}}{(d^{N_{A}+2n}+k-1)!/(d^{N_{A}+2n}-1)!},\label{SM:limiting_distribution_MSPE_after_replica}
\end{equation}
whose physical interpretation will be investigated in detail in the next section.
We will connect this result with a known random density matrix ensemble in the next section.

\section{Deep thermalization to Generalized Hilbert-Schmidt ensemble}

\label{Sec:correction_of_large_t}
We first consider the Haar random ensemble on $N_{A}+2n$ sites. Its
k-moment is then 
\begin{equation}
\rho_{\mathrm{Haar}}^{(k)}=\frac{\sum_{g'\in\mathbb{S}_k}\rho_{N_{A}+2n,k}(g')}{(d^{N_{A}+2n}+k-1)!/(d^{N_{A}+2n}-1)!}
\end{equation}
If we trace out the last $2n$ qudits, we get the generalized Hilbert-Schmidt
ensemble, whose $k$-moment is given by 
\begin{equation}
\rho_{\mathrm{GHS}}^{(k)}=\frac{\sum_{g'\in\mathbb{S}_k}\mathrm{Tr}_{2n}^{\otimes k}\rho_{N_{A}+2n,k}(g')}{(d^{N_{A}+2n}+k-1)!/(d^{N_{A}+2n}-1)!}=\frac{\sum_{g'\in\mathbb{S}_k}\rho_{N_{A},k}(g')d^{2nl(g')}}{(d^{N_{A}+2n}+k-1)!/(d^{N_{A}+2n}-1)!},
\end{equation}
Here $\mathrm{Tr}^{\otimes k}$ indicates that the partial trace is
performed on each replica. This result indicates that $\rho_{\mathrm{GHS}}^{(k)}=\sum_{g'\in\mathbb{S}_k}\alpha(g')\rho_{N_{A},k}(g')$
with $\alpha(g')\propto d^{2nl(g')}$ which agrees with Eq.~(\ref{SM:limiting_distribution_MSPE_after_replica}).
Thus, in the leading order of $d^{t}$, i.e., $t\to\infty$, the MSPE
is deeply thermalized to the generalized Hilbert Schmidt ensemble with $2n$
qudits erased.

In the following, we consider how the the $k$-th moment of MSPE approaches the aforementioned universal ensemble at long times, in particular at the order
of $\mathcal{O}\left(\frac{1}{d^{t}}\right)$. 
For simplicity, our calculation will be restricted to the leading order contribution from $\mathcal{O}(d^{N_A})$.
Employing Eq.~(\ref{eq:EndMatter_nextleadinginGHS}) in the main text, we have
\begin{equation}
\begin{aligned}c(g) & =\frac{\sum_{g'}d^{nl(g)+nl(g')+(t+1-n)l(g^{-1}g')}}{d^{k(t+1+n)}}
 & =\frac{d^{2nl(g)}+\sum_{\langle i,j\rangle}d^{nl(g)+nl(gs_{i,j})}\frac{1}{d^{t+1-n}}}{d^{2kn}}+\mathcal{O}\left(\frac{1}{d^{2t}}\right),
\end{aligned}
\end{equation}
with $s_{i,j}$ the swap operator between the $i$-th replica and $j$-th replica.
We also have
\begin{equation}
\begin{aligned}c(g) & =\sum_{g'}\alpha(g')\frac{d^{(t+1)l(g^{-1}g')}}{d^{(t+1)k}}
 & =\alpha(g)+\sum_{\langle i,j\rangle}\alpha(gs_{i,j})\frac{1}{d^{t+1}}+\mathcal{O}\left(\frac{1}{d^{2t}}\right).
\end{aligned}
\end{equation}
Solving the above equations, we obtain that 
\begin{equation}
\alpha(g)=\frac{d^{2nl(g)}}{d^{2kn}}+\frac{1}{d^{2kn+t+1}}\sum_{\langle i,j\rangle}\bigg(d^{nl(g)+nl(gs_{i,j})+n}-d^{2nl(gs_{i,j})}\bigg)+\mathcal{O}\left(\frac{1}{d^{2t}}\right).
\end{equation}
Thus, we have 
\begin{equation}
\tilde{\rho}_{N_{A}}^{(k)}=\sum_{g}\frac{\alpha(g)}{C}\rho_{N_{A},k}(g)
\end{equation}
where $C$ is a normalization factor determined by $\mathrm{Tr}\tilde{\rho}_{N_{A}}^{(k)}=1$.
Using the identity that 
\begin{equation}
\sum_{g}D^{l(g)}=\frac{(D+k-1)!}{(D-1)!}\ \mathrm{for}\ \forall D\in\mathbb{N}, \label{eq:SM_sum_of_permutation_element_inexp}
\end{equation}
we can obtain 
\begin{equation}
C=\frac{1}{d^{2kn}}\bigg(\frac{(d^{2n+N_{A}}+k-1)!}{(d^{2n+N_{A}}-1)!}+\frac{K}{d^{t+1}}\bigg)+\mathcal{O}\left(\frac{1}{d^{2t}}\right), \quad K=\sum_{g}\sum_{\langle i,j\rangle}d^{N_{A} l(g)}\bigg(d^{n l(g)+n l(gs_{i,j})+n}-d^{2nl(gs_{i,j})}\bigg).
\end{equation}
Thus, we got 
\begin{equation}
\frac{\alpha(g)}{C}=\frac{d^{2nl(g)}}{(d^{2n+N_{A}}+k-1)!/(d^{2n+N_{A}}-1)!}+\beta(g),
\end{equation}
where 
\begin{align}
\beta(g) & =\frac{(d^{2n+N_{A}}-1)!}{(d^{2n+N_{A}}+k-1)!d^{t+1}}\{\sum_{\langle i,j\rangle}d^{nl(g)+nl(gs_{ij})+n}-d^{2nl(gs_{i,j})}\nonumber 
  -K d^{2nl(g)}\frac{(d^{2n+N_{A}}-1)!}{(d^{2n+N_{A}}+k-1)!}\}+\mathcal{O}\left(\frac{1}{d^{2t}}\right).
\end{align}

We can now consider the large $N_{A}$ limit. In this limit, we have
\begin{equation}
\begin{aligned}K & =d^{kN_{A}}\sum_{\langle i,j\rangle}\bigg(d^{kn+nl(s_{i,j})+n}-d^{2nl(s_{i,j})}+\mathcal{O}\left(\frac{1}{d^{N_{A}}}\right)\bigg)
  =d^{kN_{A}}\frac{k(k-1)}{2}d^{2kn}\bigg(1-\frac{1}{d^{2n}}+\mathcal{O}\left(\frac{1}{d^{N_{A}}}\right)\bigg).
\end{aligned}
\end{equation}
and 
\begin{equation}
\begin{aligned}\beta(g) & =\frac{1}{d^{t+1}d^{2kn+kN_{A}}(1+\mathcal{O}\left(\frac{1}{d^{N_{A}}}\right))}\{\sum_{\langle i,j\rangle}d^{nl(g)+nl(gs_{ij})+n}-d^{2nl(gs_{i,j})}\\
 & -d^{kN_{A}}\frac{k(k-1)}{2}d^{2kn}\bigg(1-\frac{1}{d^{2n}}+\mathcal{O}\left(\frac{1}{d^{N_{A}}}\right)\bigg)\frac{d^{2nl(g)}}{d^{2nk+kN_{A}}(1+\mathcal{O}\left(\frac{1}{d^{N_{A}}}\right))}\}\\
 & =\frac{1}{d^{t+1}d^{2kn+kN_{A}}}\{\sum_{\langle i,j\rangle}d^{nl(g)+nl(gs_{ij})+n}-d^{2nl(gs_{i,j})} -\frac{k(k-1)}{2}\bigg(1-\frac{1}{d^{2n}}\bigg)d^{2nl(g)}+\mathcal{O}\left(\frac{1}{d^{N_{A}}}\right)\}.
\end{aligned}\label{eq:SM:expression_of_beta_g}
\end{equation}

The convergence rate can be characterized as 
\begin{equation}
\Delta^{(k)}_{\xi}=\frac{||\rho_{N_{A}}^{(k)}-\rho_{\mathrm{GHS}}^{(k)}||_{\xi}}{||\rho_{\mathrm{GHS}}^{(k)}||_{\xi}}.
\end{equation}
To compute this quantity, we also refer to the replica trick. We will first compute 
\begin{equation}
\bar{\Delta}^{(k,q)}_{\xi}=\frac{||\bar{\rho}_{N_{A}}^{(k,q)}-\rho_{\mathrm{GHS}}^{(k)}||_{\xi}}{||\rho_{\mathrm{GHS}}^{(k)}||_{\xi}}
\end{equation}
for each positive integer $q$, and then take the limit $q=1-k$. This limit exists because we will show that $\bar{\Delta}^{(k,q)}_{\xi}$ can be bounded independently of $q$. 

From the definition of $\bar{\rho}_{N_A}^{(k,q)}$ Eq. (\ref{SM:expression_of_bar_rho_k_q}), we obtain that 
\begin{equation}
\bar{\Delta}^{(k,q)}_{\xi}=\frac{||\sum_{g\in\mathbb{S}_{k+q}}\beta(g)\mathrm{Tr}_{\{k+q\}/\{k\}}\rho_{N_{A},k+q}(g)||_{\xi}}{||\rho_{\mathrm{GHS}}^{(k)}||_{\xi}},
\end{equation}
with $\beta(g)$ given by Eq. (\ref{eq:SM:expression_of_beta_g}) but with $k$ replaced by $k+q$. Since we only consider the leading contribution from $\mathcal{O}(d^{N_A})$, the trace $\mathrm{Tr}_{\{k+q\}/\{k\}}$ forces that $\rho_{N_A,k+q}(g)=\rho_{N_A,k}(g')\otimes \rho_{N_A,q}(e)$ for some $g'\in\mathbb{S}_k$. Equivalently, it forces that the permutation element $g$ should act like $g'$ in the first $k$ elements and act trivially in the last $q$ elements, which we denote by $g=(g', e)$. Using $\mathrm{Tr}_{\{k+q\}/\{k\}}\rho_{N_A,k+q}(g)=\rho_{N_A}(g')d^{qN_A}$, we have
\begin{equation}
\overline{\Delta}_{\xi}^{(k,q)}=\frac{\left\lVert \sum_{g=(g',e)}\bigg\{\sum_{1\leq i<j\leq k+q}d^{nl(g)+nl(gs_{ij})+n}-d^{2nl(gs_{i,j})}-\frac{(k+q)(k+q-1)}{2}\bigg(1-\frac{1}{d^{2n}}\bigg)d^{2nl(g)}\bigg\} d^{qN_{A}}\rho_{N_{A},k}(g')\right\rVert _{\xi}}{d^{t+1}d^{2(k+q)n+(k+q)N_{A}}\lVert\rho_{\mathrm{GHS}}^{(k)}\rVert_{\xi}}
\end{equation}
up to the leading contribution of $\mathcal{O}(d^{N_A})$.

We first consider the term
\[
\sum_{1\leq i<j\leq k+q}d^{nl(g)+nl(gs_{ij})+n}-d^{2nl(gs_{i,j})}-\frac{(k+q)(k+q-1)}{2}\bigg(1-\frac{1}{d^{2n}}\bigg)d^{2nl(g)}
\]
in the numerator. Since $g=(g',e)$, we have $l(g)=l(g')+q$. We classify
the sum into three groups. 
\begin{enumerate}
\item $1\leq i<j\leq k$. This group has $k(k-1)/2$ elements and $l(gs_{i,j})=l(g's_{i,j})+q$.
This leads to
\[
\begin{aligned} & \sum_{1\leq i<j\leq k}d^{nl(g)+nl(gs_{ij})+n}-d^{2nl(gs_{i,j})}-\frac{k(k-1)}{2}\bigg(1-\frac{1}{d^{2n}}\bigg)d^{2nl(g)}\\
= & \left[\sum_{1\leq i<j\leq k}d^{nl(g')+nl(g's_{ij})+n}-d^{2nl(g's_{i,j})}-\frac{k(k-1)}{2}\bigg(1-\frac{1}{d^{2n}}\bigg)d^{2nl(g')}\right]d^{2nq};
\end{aligned}
\]
\item $1\leq i\leq k<j\leq k+q$. This group has $kq$ elements and $l(gs_{i,j})=l(g')+q-1$.
Therefore,
\[
\begin{aligned} & \sum_{1\leq i\leq k<j\leq k+q}d^{nl(g)+nl(gs_{ij})+n}-d^{2nl(gs_{i,j})}-kq\bigg(1-\frac{1}{d^{2n}}\bigg)d^{2nl(g)}\\
= & kq\left[d^{2n(l(g')+q)}-d^{2n(l(g')+q-1)}-\left(1-\frac{1}{d^{2n}}\right)d^{2n(l(g')+q)}\right]\\
= & 0;
\end{aligned}
\]
\item $k+1\leq i<j\leq k+q$. This group has $q(q-1)/2$ elements and $l(gs_{i,j})=l(g')+q-1$.
Similarly
\[
\begin{aligned} & \sum_{k+1\leq i<j\leq k+1}d^{nl(g)+nl(gs_{ij})+n}-d^{2nl(gs_{i,j})}-kq\bigg(1-\frac{1}{d^{2n}}\bigg)d^{2nl(g)}= 0.
\end{aligned}
\]
\end{enumerate}
Summing over the above three groups, we obtain
\begin{equation}
\overline{\Delta}_{\xi}^{(k,q)}=\frac{\left\lVert \sum_{g'\in\mathbb{S}_k}\bigg\{\sum_{1\leq i<j\leq k}d^{nl(g')+nl(g's_{ij})+n}-d^{2nl(g's_{i,j})}-\frac{k(k-1)}{2}\bigg(1-\frac{1}{d^{2n}}\bigg)d^{2nl(g')}\bigg\}\rho_{N_{A},k}(g')\right\rVert _{\xi}}{d^{t+1}d^{2kn+kN_{A}}\lVert\rho_{\mathrm{GHS}}^{(k)}\rVert_{\xi}}.\label{eq:SM_deviation_onenorm}
\end{equation}
The above expression is independent of $q$. Therefore, we can take the limit $q=1-k$ and the same expression holds for $\Delta_{\xi}^{(k)}$.

For the case $\xi=1$, let us estimate the numerator in Eq. (\ref{eq:SM_deviation_onenorm}).
It is obvious that $\rho_{N_{A},k}(g')\rho_{N_{A},k}^{\dagger}(g')=I$,
where $I$ is the Identity acting on the $k$-replica space. Thus,
$||\rho_{N_{A}}(g')||_{1}=d^{N_{A}k}$.
For the denominator, we have
$
||\rho_{\mathrm{GHS}}^{(k)}||_{1}=1$
since $\rho_{\mathrm{GHS}}^{(k)}$ is a normalized density matrix. 
Therefore,
\begin{equation}
\begin{aligned}\Delta^{(k)}_{1} & \leq\frac{d^{N_{A}k}}{d^{t+1}d^{2kn+kN_{A}}}\{\sum_{g\in\mathbb{S}_k}|\sum_{\langle i,j\rangle}d^{nl(g)+nl(gs_{ij})+n}-d^{2nl(gs_{i,j})}-\frac{k(k-1)}{2}\bigg(1-\frac{1}{d^{2n}}\bigg)d^{2nl(g)}|\}\\
 & =\frac{1}{d^{t+1}d^{2kn}}\{\sum_{g\in\mathbb{S}_k}|\sum_{\langle i,j\rangle}d^{nl(g)+nl(gs_{ij})+n}-d^{2nl(gs_{i,j})}-\frac{k(k-1)}{2}\bigg(1-\frac{1}{d^{2n}}\bigg)d^{2nl(g)}|\}
\end{aligned}
\label{eq:convergence_rate_for_2-1}
\end{equation}
If  the bound for $\Delta_1^{(k)}$ converges to zero, it 
 means that terms vanish for each
$g$ in the sum. If we take $g=e$, the summand reduces to zero. 

Next, we consider the contribution from $g=s_{i',j'}$ for certain
pair of $(i',j')$ with $1\leq i'<j'\leq k$, which gives us

\begin{equation}
\begin{aligned} & \frac{1}{d^{t+1}d^{2kn}}|\{d^{2nk}+\frac{(k-2)(k+1)}{2}d^{2n(k-1)}-d^{2nk}- \frac{(k-2)(k+1)}{2}d^{2n(k-2)}-\frac{k(k-1)}{2}\bigg(1-\frac{1}{d^{2n}}\bigg)d^{2n(k-1)}\}|\\
= & \frac{1}{d^{t+1}d^{2kn}}|\{\frac{(k-2)(k+1)}{2}d^{2n(k-1)}(1-\frac{1}{d^{2n}})-\frac{k(k-1)}{2}\bigg(1-\frac{1}{d^{2n}}\bigg)d^{2n(k-1)}\}|\\
= & \frac{1}{d^{t+1}d^{2kn}}d^{2n(k-1)}(1-\frac{1}{d^{2n}})
\sim  \frac{1}{d^{t+1}}\frac{1}{d^{2n}}(1-\frac{1}{d^{2n}}).
\end{aligned}
\end{equation}
Summing over all pairs of $(i',j')$ with $1\leq i'<j'\leq k$ gives rise to 
the leading contribution in the limit $n\to\infty$ while the ratio $n/N_A$ is fixed (or equivalently, in the large $d$ limit).  Thus, we obtain that in this case,
when $k\ll d^{2n}$,
 
\begin{equation}
\Delta^{(k)}_1\lesssim\frac{1}{d^{t+1}}\frac{1}{d^{2n}}(1-\frac{1}{d^{2n}})\frac{k(k-1)}{2}.
\end{equation}
If we require that $\Delta^{(k)}_1\leq\epsilon$, we got the time it
takes for MSPE to approach the generalized Hilbert-Schmidt ensemble $\epsilon$
closely as $t_{k}\lesssim\log{k/\epsilon}$. We have also numerically tested that the dominant contribution of $\Delta_{1}^{(k)}$
 comes from summing $g=s_{i',j'}$ for all pairs of $(i',j')$ even for finite $n$ when $k$ is not
too large compared to $d^{2n}$.

The case $\xi=2$ follows similarly with being more controllable. 
Using Eq.~(\ref{eq:SM_sum_of_permutation_element_inexp}), we arrive at

\begin{equation}
\begin{aligned}\Delta_{2}^{(k)} & =\bigg(\frac{\sum_{g}|\beta(g)|^{2}d^{N_{A}k}
}{\sum_{g}|\frac{d^{2nl(g)}}{(d^{N_{A}+2n}+k-1)!/(d^{N_{A}+2n}-1)!}|^{2}d^{N_{A}k}
}
\bigg)^{\frac{1}{2}}
\\
 & =\frac{1}{d^{t+1}}\frac{\frac{1}{d^{2kn+kN_{A}}}\bigg(\sum_{g}|\{\sum_{\langle i,j\rangle}d^{nl(g)+nl(gs_{ij})+n}-d^{2nl(gs_{i,j})}-\frac{k(k-1)}{2}\bigg(1-\frac{1}{d^{2n}}\bigg)d^{2nl(g)}+\mathcal{O}\left(\frac{1}{d^{N_{A}}}\right)\}|^{2}\bigg)^{\frac{1}{2}}}{\frac{1}{d^{kN_{A}+2kn}
 }
 \bigg(\sum_{g}|d^{2nl(g)}|^{2}\bigg)^{\frac{1}{2}}}
 \\
 & =\frac{1}{d^{t+1}}\bigg(\frac{(d^{4n}-1)!}{(d^{4n}+k-1)!}\sum_{g}|\{\sum_{\langle i,j\rangle}d^{nl(g)+nl(gs_{ij})+n}-d^{2nl(gs_{i,j})}-\frac{k(k-1)}{2}\bigg(1-\frac{1}{d^{2n}}\bigg)d^{2nl(g)}\}|^{2}\bigg)^{\frac{1}{2}}
 \\
 &\sim
 \frac{1}{d^{t+1}}\frac{1}{d^{2n}}(1-\frac{1}{d^{2n}})\sqrt{\frac{k(k-1)}{2}}.
\end{aligned}
\label{eq:convergence_rate_for_norm_2}
\end{equation}
Notice that in the case $\xi=1$, we obtain an upper bound on $\Delta_1^{(k)}$, but for $\xi=2$, the result is nearly precise except the last line in Eq. (\ref{eq:convergence_rate_for_norm_2}). If we require that $\Delta_2^{(k)}\leq\epsilon$, we got the time it takes for MSPE to approach the generalized Hilbert-Schmidt ensemble $\epsilon$ closely as $t_k\sim\log{k/\epsilon}$.

In the above discussion, we assume $m=2n$ as a constant. Nonetheless,
one can directly figure out that all of the results in this section
also hold if $N_{A}\to\infty$ and the ratio $\gamma=m/N_{A}$ is fixed ($\gamma=0$ corresponds to a constant $m$).

\section{Large $d$ limit}
\label{Sec:large_d}
In this section, we take the limit $d\to\infty$ at fixed $t$ for any $t+1\geq N_A$. In this limit, we can approximate the cycle function $l(g)$ for $g\in\mathbb{S}_k$ as
\begin{equation}
    d^{l(g)}=d^k\bigg(\delta_{g,e}+\mathcal{O}\left(\frac{1}{d}\right)\bigg).
\end{equation}
Thus, at the zeroth order, we have $\rho_{N_A}^{(k)}=\sum_{g\in\mathbb{S}_k} \alpha(g) \rho_{N_A,k}(g)/\mathrm{Tr}(\sum_{g\in\mathbb{S}_k} \alpha(g) \rho_{N_A,k}(g))$ with
\begin{equation}
    \alpha(g)=
\begin{cases}
       1& \; \mathrm{if}\;g=e\\
       0& \; \mathrm{otherwise}
\end{cases}.
\end{equation}
We want to emphasize that this result (at fixed $k$) holds for any $n\geq1$. It indicates that in the large-$d$ limit, the contribution from the permutation element except the Identity one diminishes sharply even for a single forgotten measurement outcome. The MSPE deeply thermalizes to the totally mixed states.

In the next order expansion in $1/d$, we can approximate the cycle function $l(g)$ as 
\begin{equation}
    d^{l(g)}=d^k\bigg(\delta_{g,e}+\frac{1}{d}\sum_{i<j}\delta_{g,s_{i,j}}+\mathcal{O}\left(\frac{1}{d^2}\right)\bigg).
\end{equation}
If $t>3n$, we can still do the approximation 
\begin{equation}
d^{(t+1-n)l(g^{-1}g')}=d^{(t+1-n)k}\delta_{g,g'}+\mathcal{O}\left(\frac{1}{d^{t-n}}\right),
\end{equation}
and similarly for $d^{(t+1)l(g^{-1}g')}$.
Consequently, by repeating the argument above again, we obtain that
\begin{equation}
\alpha(g)=\delta_{g,e}+\sum_{i<j}\delta_{g,s_{i,j}}\frac{1}{d^{2n}}+\mathcal{O}\left(\frac{1}{d^{2n+1}}\right).
\end{equation}
Therefore, we see the effect of $n$ on MSPE. 
Although the amplitude of each term in the summand is independent of $k$, the number of terms in $\sum_{i<j}$ scales as $k^2$. Consequently, we can not ignore the next leading contribution if $k\gg d^n$.

\section{Eigenvalue distribution}
\label{Sec:Eigenvalue_distribution}

In this section, we compute the eigenvalue distribution for MSPE.
From the  section~\ref{Sec:correction_of_large_t}, we can express the $k$-moment of MSPE
as 
\begin{equation}
\rho_{N_{A}}^{(k)}=\sum_{g\in\mathbb{S}_k}\frac{\alpha(g)}{C}\rho_{N_{A},k}(g)\label{eq:SM_kth_moment_of_MSPE}
\end{equation}
with 
\begin{equation}
\begin{aligned}\frac{\alpha(g)}{C} & =\frac{d^{ml(g)}}{(d^{m+N_{A}}+k-1)!/(d^{m+N_{A}}-1)!}\\
 & +\frac{1}{d^{t+1}d^{mk+kN_{A}}}\bigg\{\sum_{1\leq i<j\leq k}d^{\frac{ml(g)+ml(gs_{ij})+m}{2}}-d^{ml(gs_{i,j})}-\frac{k(k-1)}{2}\bigg(1-\frac{1}{d^{m}}\bigg)d^{ml(g)}\bigg\}.
\end{aligned}
\end{equation}
Here we have substituted $m=2n$. In this section, we take the limit
$N_{A}\to\infty$ with the ratio $\gamma=m/N_{A}$ fixed ($\gamma$
can be zero, which corresponds to $m$ fixed as a constant). We first
need to calculate $\mathrm{Tr}\overline{\rho_{N_{A}}^{k}}=\sum_{\bm{\alpha}}P_{\bm{\alpha}}\mathrm{Tr}\rho_{N_A}^k(\bm{\alpha})$, where
we use bar to denote the average over MSPE. To this end, we introduce
a permutation element $\sigma$ which maps $\{1,2,3,\cdots,k\}$ to
$\{k,1,2,\cdots,k-1\}$. 
Denoting by $\mathcal{S}=\rho_{N_{A},k}(\sigma)$ as the representation of the permutation element $\sigma$ on $k$-replicas, each with $N_A$ sites, and denoting by $\rho_{N_A}$ the usual density matrix of a single replica,
we obtain 
\begin{equation}
\mathrm{Tr}\overline{\rho_{N_{A}}^{k}}=\mathrm{Tr}\mathcal{S}^\dagger \rho_{N_{A}}^{(k)}=\min\{d^{N_{A}},d^{m}\}\int P(\lambda)\lambda^{k}d\lambda.\label{eq:SM_expression_of_distribution_EV}
\end{equation}
In the above expression, $P(\lambda)$ is the distribution of the
eigenvalues, and the prefactor $\min\{d^{N_{A}},d^{m}\}$ is the maximum possible number
of non-vanishing eigenvalues of $\rho_{N_{A}}(\bm{\alpha})$. 

As long as the left-hand side of Eq. (\ref{eq:SM_expression_of_distribution_EV})
is known, we can then revert the right hand side and obtain the expression
of $P(\lambda)$. In the following, we will only keep the leading
order of $d^{N_{A}}$. If $\gamma<1$, we have
\begin{equation}
\begin{aligned}  \mathrm{Tr}\overline{\rho_{N_{A}}^{k}}
= & \sum_{g}\frac{\alpha(g)}{C}d^{N_{A}l(\sigma^{-1}g)}
\approx  \frac{\alpha(\sigma)}{C}d^{kN_{A}}
=  \frac{1}{d^{(k-1)m}}-\frac{k(k-1)}{2d^{t+1}}\bigg(1-\frac{1}{d^{m}}\bigg)\frac{1}{d^{(k-1)m}}.
\end{aligned}
\end{equation}
This leads to the solution of $P(\lambda)$ as 
\begin{equation}
P(\lambda)=\delta\left(\lambda-\frac{1}{d^{m}}\right)-\frac{1}{2d^{t+1}d^{2m}}\bigg(1-\frac{1}{d^{m}}\bigg)\delta''\left(\lambda-\frac{1}{d^{m}}\right),\label{eq:SM_P_lambda_less}
\end{equation}
where $\delta''(x)$ is the second derivative of $\delta(x)$. The
first term in Eq. (\ref{eq:SM_P_lambda_less}) is the same result
as that of the generalized Hilbert-Schmidt ensemble in the limit of $m, N_A\to\infty$ with a fixed ratio $\gamma=m/N_A$. The second term, which is exponentially suppressed in
$t$, represents the broadening of the delta distribution, as can
be seen by a test function.

The case $\gamma>1$ follows similarly. One just needs to notice that the main contribution
to $\mathrm{Tr}\overline{\rho_{N_{A}}^{k}}$ now comes from $g=e$ in
Eq. (\ref{eq:SM_kth_moment_of_MSPE}), where $\alpha(g)$ dominates.

\section{Deep thermalization for typical dynamical models\label{Sec:DT_typical}}
In this section, we will compute the MSPE both analytically and numerically for other typical models. We present the analytical result first. 
\subsection{Analytical proof for globally Haar-random states}
Consider a circuit uniformly drawn from the global unitary group; then the output state will be Haar distributed. For this setting, we establish the following
\begin{theorem}[Global Haar random states]\label{Theorem:global_haar}

Let $\ket{\psi}$ be a Haar-random state in $\mathbb{C}^{d^{N_A+N_B}}$,
subjected to a similar measurement-forgetting scheme as in Fig.~\ref{fig1} in the main text, with measurements taken in the computational basis. Then the corresponding MSPE is $\epsilon$-approximately deeply
thermalized to the generalized Hilbert-Schmidt ensemble with probability
at least $1-\delta$ as long as 
\begin{equation}
N_{B}=\Omega(N_{A}k+\log\frac{1}{\epsilon}+\log\log\frac{1}{\delta}).\label{eq:SM_theorem_scaling_of_Nb}
\end{equation}
Here the $\epsilon$-approximation is characterized in $\Delta_{1}^{(k)}$.

\end{theorem}

Notice that the required scaling of the subsystem $B$, the right-hand side of Eq. (\ref{eq:SM_theorem_scaling_of_Nb}), only depends
on $N_{A}$ but not $N_{B_{2}}$. Thus, the result also applies under dense measurement-loss errors, and the precise locations of such errors are irrelevant. In the theorem, we assume the measurement is in the computational basis, but it can be directly generalized to other measurement bases.
Our proof strategy closely parallels that of Ref. \cite{cotler2023emergent}.

The proof of the above theorem is
divided into the following three lemmas. Notice that if $\ket{\psi}\in\mathbb{C}^{d^{N_{A}+N_{B}}}$
is Haar distributed, then real vector $\vec{v}$ constructed from
$\ket{\psi}$ as $\vec{v}=\begin{bmatrix}\mathrm{Re}(\psi) & \mathrm{Im}(\psi)\end{bmatrix}$is
uniformly distributed on the unit sphere in $\mathbb{R}^{2d^{N_{A}+N_{B}}}$.
In what follows, we identify a function $f:\mathbb{C}^{D}\to\mathbb{R}$
with the corresponding function $f:\mathbb{S}^{2D-1}\to\mathbb{R}$
via the relation $f(\ket{\psi})=f(\vec{v})$. The Levy's lemma then
states that

\begin{lemma}[The Levy's lemma]\label{Lemma:Levy}

If a function $f:\mathbb{S}^{2D-1}\to\mathbb{R}$ is differentiable
and satisfies $\lVert df(\vec{v})/d\vec{v}\rVert\leq\eta$ for $\forall\vec{v}\in\mathbb{S}^{2D-1}$,
then for any $\epsilon>0$, we have
\begin{equation}
\mathrm{Prob}_{\ket{\psi}\sim\mathrm{Haar}(D)}[|f(\ket{\psi})-\mathbb{E}_{\ket{\psi}\sim\mathrm{Haar}(D)}(f(\ket{\psi}))|\geq\epsilon]\leq2\exp\left(-\frac{2D\epsilon^{2}}{9\pi^{3}\eta^{2}}\right)
\end{equation}
with $\mathbb{E}_{\ket{\psi}\sim\mathrm{Haar}(D)}$ denoting the average
over the Haar distribution.

\end{lemma}

To employ the Levy's lemma, we choose the function $f:\mathbb{C}^{d^{N_{A}+N_{B}}}\to\mathbb{R}$
as 
\begin{equation}
f(\ket{\psi})=\bra{i}\sum_{\bm{\alpha}}\frac{\mathrm{Tr}_{B_{2}}\left(\ket{\tilde{\psi}_{\bm{\alpha}}}\bra{\tilde{\psi}_{\bm{\alpha}}}\right)^{\otimes k}}{\left(\braket{\tilde{\psi}_{\bm{\alpha}}|\tilde{\psi}_{\bm{\alpha}}}\right)^{k-1}}\ket{j},\label{eq:SM_definition_of_test_function_in_Levy}
\end{equation}
where $\ket{\tilde{\psi}_{\bm{\alpha}}}=(I_{A+B_{2}}\otimes\bra{\bm{\alpha}})\ket{\psi}$
with $\bm{\alpha}$ the measurement outcome on $B_{1},$ $B_{3}$,
and $\ket{i},\ket{j}$ being the standard basis on $\mathcal{H}_{A}^{\otimes k}$,
where $\mathcal{H}_{A}$ is the Hilbert space on the subsystem $A$,
as $\ket{i}=\otimes_{l=1}^{k}\ket{i_{l}}$, $\ket{j}=\otimes_{l=1}^{k}\ket{j_{l}}$
with $\ket{i_{l}},\ket{j_{l}}\in\{\ket{1},\ket{2},\cdots,\ket{d}\}^{\otimes N_{A}}$.
The following Lemma shows that the function $f(\ket{\psi})$ satisfies
the condition in the Levy's lemma as

\begin{lemma}[Coefficients in Levy's Lemma]

The function $f$ defined in Eq. (\ref{eq:SM_definition_of_test_function_in_Levy})
satisfies 
\begin{equation}
\left\lVert \frac{df(\vec{v})}{d\vec{v}}\right\rVert \leq2(2k-1).
\end{equation}

\end{lemma}

\begin{proof}

We can explicitly calculate the derivative as 
\begin{equation}
\begin{aligned}\frac{df(\vec{v})}{d\vec{v}} & =-(k-1)\sum_{\bm{\alpha}}\left(\bra{i}\frac{\mathrm{Tr}_{B_{2}}\left(\ket{\tilde{\psi}_{\bm{\alpha}}}\bra{\tilde{\psi}_{\bm{\alpha}}}\right)^{\otimes k}}{\left(\braket{\tilde{\psi}_{\bm{\alpha}}|\tilde{\psi}_{\bm{\alpha}}}\right)^{k}}\ket{j}\frac{d}{d\vec{v}}\braket{\tilde{\psi}_{\bm{\alpha}}|\tilde{\psi}_{\bm{\alpha}}}\right)\\
 & +\sum_{\bm{\alpha}}\left(\sum_{l=1}^{k}\prod_{l'\neq l}\frac{\bra{i_{l'}}\mathrm{Tr}_{B_{2}}\left(\ket{\tilde{\psi}_{\bm{\alpha}}}\bra{\tilde{\psi}_{\bm{\alpha}}}\right)\ket{j_{l'}}}{\braket{\tilde{\psi}_{\bm{\alpha}}|\tilde{\psi}_{\bm{\alpha}}}}\frac{d}{d\vec{v}}\bra{i_{l}}\mathrm{Tr}_{B_{2}}\left(\ket{\tilde{\psi}_{\bm{\alpha}}}\bra{\tilde{\psi}_{\bm{\alpha}}}\right)\ket{j_{l}}\right).
\end{aligned}
\label{eq:SM_explicitly_derivative_of_f}
\end{equation}
We are going to upper bound each term separately. For the first term
in Eq. (\ref{eq:SM_explicitly_derivative_of_f}), we have 
\begin{equation}
\frac{d}{d\vec{v}}\braket{\tilde{\psi}_{\bm{\alpha}}|\tilde{\psi}_{\bm{\alpha}}}=2\ket{\bm{\alpha}}_{B_{1},B_{3}}\bra{\bm{\alpha}}\otimes I_{A+B_{2}}\otimes\begin{bmatrix}1\\
 & 1
\end{bmatrix}\cdot\vec{v}.
\end{equation}
 Therefore, the vector norm of the first term in Eq. (\ref{eq:SM_explicitly_derivative_of_f})
can be bounded as 
\begin{equation}
\begin{aligned} & \left\lVert -(k-1)\bra{i}\frac{\mathrm{Tr}_{B_{2}}\left(\ket{\tilde{\psi}_{\bm{\alpha}}}\bra{\tilde{\psi}_{\bm{\alpha}}}\right)^{\otimes k}}{\left(\braket{\tilde{\psi}_{\bm{\alpha}}|\tilde{\psi}_{\bm{\alpha}}}\right)^{k}}\ket{j}\frac{d}{d\vec{v}}\braket{\tilde{\psi}_{\bm{\alpha}}|\tilde{\psi}_{\bm{\alpha}}}\right\rVert \\
 & \leq2(k-1)\sqrt{\sum_{\bm{\alpha}}\bra{i}\frac{\mathrm{Tr}_{B_{2}}\left(\ket{\tilde{\psi}_{\bm{\alpha}}}\bra{\tilde{\psi}_{\bm{\alpha}}}\right)^{\otimes k}}{\left(\braket{\tilde{\psi}_{\bm{\alpha}}|\tilde{\psi}_{\bm{\alpha}}}\right)^{k}}\ket{j}\vec{v}^{T}\ket{\bm{\alpha}}_{B_{1},B_{3}}\bra{\bm{\alpha}}\otimes I_{A+B_{2}}\otimes\begin{bmatrix}1\\
 & 1
\end{bmatrix}\vec{v}},\\
 & \leq2(k-1)\sqrt{\sum_{\bm{\alpha}}\vec{v}^{T}\ket{\bm{\alpha}}_{B_{1},B_{3}}\bra{\bm{\alpha}}\otimes I_{A+B_{2}}\otimes\begin{bmatrix}1\\
 & 1
\end{bmatrix}\vec{v}},\\
 & \leq2(k-1),
\end{aligned}
\label{eq:SM_bound_result_first}
\end{equation}
where in the first inequality, we used $P_{\bm{\alpha}}P_{\bm{\alpha'}}=P_{\bm{\alpha}}\delta_{\bm{\alpha,\alpha'}}$,
and in the second inequality, we used the fact that $\mathrm{Tr}_{B_{2}}(\ket{\tilde{\psi}_{\bm{\alpha}}}\bra{\tilde{\psi}_{\bm{\alpha}}})^{\otimes k}/(\braket{\tilde{\psi}_{\bm{\alpha}}|\tilde{\psi}_{\bm{\alpha}}})^{k}$
is a density matrix, such that its operator norm is less or equal
than $1$. Similarly, we can upper bound the second term in Eq. (\ref{eq:SM_explicitly_derivative_of_f})
as 
\begin{equation}
\begin{aligned} & \frac{d}{d\vec{v}}\bra{i_{l}}\mathrm{Tr}_{B_{2}}\left(\ket{\tilde{\psi}_{\bm{\alpha}}}\bra{\tilde{\psi}_{\bm{\alpha}}}\right)\ket{j_{l}}\\
 & =\sum_{\bm{\beta}}\frac{d}{d\vec{v}}\bra{i_{l}}\bra{\bm{\beta}}_{B_{2}}\left(\ket{\tilde{\psi}_{\bm{\alpha}}}\bra{\tilde{\psi}_{\bm{\alpha}}}\right)\ket{\bm{\beta}}_{B_{2}}\ket{j_{l}},\\
 & =\sum_{\bm{\beta}}\ket{\bm{\alpha}}_{B_{1},B_{3}}\bra{\bm{\alpha}}\otimes\ket{\bm{\beta}}_{B_{2}}\bra{\bm{\beta}}\otimes\begin{bmatrix}\ket{j_{l}}_{A}\bra{i_{l}}+\ket{i_{l}}_{A}\bra{j_{l}} & i\ket{j_{l}}_{A}\bra{i_{l}}-i\ket{i_{l}}_{A}\bra{j_{l}}\\
-i\ket{j_{l}}_{A}\bra{i_{l}}+i\ket{i_{l}}_{A}\bra{j_{l}} & \ket{j_{l}}_{A}\bra{i_{l}}+\ket{i_{l}}_{A}\bra{j_{l}}
\end{bmatrix}\cdot\vec{v},
\end{aligned}
\end{equation}
where $\{\ket{\bm{\beta}}_{B_{2}}\}$ is a set of orthogonal basis
on the region $B_{2}$. To simplify the notations, we introduce 
\begin{equation}
M_{l}(\bm{\alpha},\bm{\beta})=\ket{\bm{\alpha}}_{B_{1},B_{3}}\bra{\bm{\alpha}}\otimes\ket{\bm{\beta}}_{B_{2}}\bra{\bm{\beta}}\otimes\begin{bmatrix}\ket{j_{l}}_{A}\bra{i_{l}}+\ket{i_{l}}_{A}\bra{j_{l}} & i\ket{j_{l}}_{A}\bra{i_{l}}-i\ket{i_{l}}_{A}\bra{j_{l}}\\
-i\ket{j_{l}}_{A}\bra{i_{l}}+i\ket{i_{l}}_{A}\bra{j_{l}} & \ket{j_{l}}_{A}\bra{i_{l}}+\ket{i_{l}}_{A}\bra{j_{l}}
\end{bmatrix}\label{eq:SM_definition_of_MLalphabeta}
\end{equation}
and 
\begin{equation}
C_{l}(\bm{\alpha})=\prod_{l'\neq l}\frac{\bra{i}_{l'}\mathrm{Tr}_{B_{2}}\left(\ket{\tilde{\psi}_{\bm{\alpha}}}\bra{\tilde{\psi}_{\bm{\alpha}}}\right)\ket{j}_{l'}}{\braket{\tilde{\psi}_{\bm{\alpha}}|\tilde{\psi}_{\bm{\alpha}}}}.
\end{equation}
By orthogonality, we have $M_{l}(\bm{\alpha},\bm{\beta})M_{l'}(\bm{\alpha}',\bm{\beta}')=0$
unless $\bm{\alpha}=\bm{\alpha}'$, and $\bm{\beta}=\bm{\beta}'$.
Since $\mathrm{Tr}_{B_{2}}\ket{\tilde{\psi}_{\bm{\alpha}}}\bra{\tilde{\psi}_{\bm{\alpha}}}/\braket{\tilde{\psi}_{\bm{\alpha}}|\tilde{\psi}_{\bm{\alpha}}}$
is a valid state, we also have
\begin{equation}
|C_{l}(\bm{\alpha})|\leq1.
\end{equation}
Therefore, the second term in Eq. (\ref{eq:SM_explicitly_derivative_of_f})
can be bounded as
\begin{equation}
\begin{aligned} & \left\lVert \sum_{\bm{\alpha}}\left(\sum_{l=1}^{k}C_{l}(\bm{\alpha})\frac{d}{d\vec{v}}\bra{i}_{l}\mathrm{Tr}_{B_{2}}\left(\ket{\tilde{\psi}_{\bm{\alpha}}}\bra{\tilde{\psi}_{\bm{\alpha}}}\right)\ket{j}_{l}\right)\right\rVert \\
 & \leq\sqrt{\sum_{\bm{\alpha}}\sum_{l=1}^{k}\sum_{\bm{\alpha}'}\sum_{l'=1}^{k}\sum_{\bm{\beta}}\sum_{\bm{\beta}'}C_{l}(\bm{\alpha})C_{l'}^{*}(\bm{\alpha'})\vec{v}^{T}M_{l'}^{\dagger}(\bm{\alpha}',\bm{\beta}')M_{l}(\bm{\alpha},\bm{\beta})\vec{v}},\\
 & \leq\sqrt{\sum_{\bm{\alpha}}\sum_{\bm{\beta}}\sum_{l=1}^{k}\sum_{l'=1}^{k}\vec{v}^{T}M_{l'}^{\dagger}(\bm{\alpha},\bm{\beta})M_{l}(\bm{\alpha},\bm{\beta})\vec{v}}.
\end{aligned}
\label{eq:SM_intermediate_step_in_bounding_second_term}
\end{equation}
From the definition Eq. (\ref{eq:SM_definition_of_MLalphabeta}),
we can directly calculate that 
\begin{equation}
\begin{aligned}\sum_{\bm{\alpha}}\sum_{\bm{\beta}}M_{l'}^{\dagger}(\bm{\alpha},\bm{\beta})M_{l}(\bm{\alpha},\bm{\beta}) & =I_{B}\otimes\bigg(\begin{bmatrix}\ket{j_{l}}_{A}\bra{i_{l}}+\ket{i_{l}}_{A}\bra{j_{l}} & i\ket{j_{l}}_{A}\bra{i_{l}}-i\ket{i_{l}}_{A}\bra{j_{l}}\\
-i\ket{j_{l}}_{A}\bra{i_{l}}+i\ket{i_{l}}_{A}\bra{j_{l}} & \ket{j_{l}}_{A}\bra{i_{l}}+\ket{i_{l}}_{A}\bra{j_{l}}
\end{bmatrix}\cdot\\
 & \begin{bmatrix}\ket{j'_{l'}}_{A}\bra{i'_{l'}}+\ket{i'_{l'}}_{A}\bra{j'_{l'}} & i\ket{j'_{l'}}_{A}\bra{i'_{l'}}-i\ket{i'_{l'}}_{A}\bra{j'_{l'}}\\
-i\ket{j'_{l'}}_{A}\bra{i'_{l'}}+i\ket{i'_{l'}}_{A}\bra{j'_{l'}} & \ket{j'_{l'}}_{A}\bra{i'_{l'}}+\ket{i'_{l'}}_{A}\bra{j'_{l'}}
\end{bmatrix}\bigg),
\end{aligned}
\end{equation}
whose operator norm can be bounded by
\begin{equation}
\lVert\sum_{\bm{\alpha}}\sum_{\bm{\beta}}M_{l'}^{\dagger}(\bm{\alpha},\bm{\beta})M_{l}(\bm{\alpha},\bm{\beta})\rVert\leq4.
\end{equation}
Thus, Eq. (\ref{eq:SM_intermediate_step_in_bounding_second_term})
can be simplified to 
\begin{equation}
\left\lVert \sum_{\bm{\alpha}}\left(\sum_{l=1}^{k}C_{l}(\bm{\alpha})\frac{d}{d\vec{v}}\bra{i}_{l}\mathrm{Tr}_{B_{2}}\left(\ket{\tilde{\psi}_{\bm{\alpha}}}\bra{\tilde{\psi}_{\bm{\alpha}}}\right)\ket{j}_{l}\right)\right\rVert \leq\sqrt{4k^{2}}=2k.\label{eq:SM_bound_result_second}
\end{equation}
Combining Eqs. (\ref{eq:SM_bound_result_first}) and (\ref{eq:SM_bound_result_second}),
we obtain 
\begin{equation}
\left\lVert \frac{df(\vec{v})}{d\vec{v}}\right\rVert \leq2(2k-1).
\end{equation}

\end{proof}

The above Lemma, together with the Levy's Lemma, indicates that $f(\vec{v})$
defined in Eq. (\ref{eq:SM_definition_of_test_function_in_Levy})
will concentrate on its average value, which is given by the following
lemma:

\begin{lemma}[Avege value]The average value of $f(\vec{v})$ can
be calculated as

\begin{equation}
\mathbb{E}_{\ket{\psi}\sim\mathrm{Haar}(d^{N_{A}+N_{B}})}f(\ket{\psi})=\bra{i}\rho_{\mathrm{GHS}}^{(k)}\ket{j}.
\end{equation}

\end{lemma}

\begin{proof}

We can rewrite $f(\ket{\psi})$ as 
\begin{equation}
f(\ket{\psi})=\bra{i}\mathrm{Tr}_{B_{2}}\left(\sum_{\bm{\alpha}}P_{\bm{\alpha}}(\ket{\psi_{\bm{\alpha}}}\bra{\psi_{\bm{\alpha}}})^{\otimes k}\right)\ket{j},
\end{equation}
where $P_{\bm{\alpha}}=\braket{\psi|I_{A+B_{2}}\otimes\ket{\bm{\alpha}}\bra{\bm{\alpha}}|\psi}$
and $\ket{\psi_{\bm{\alpha}}}=\ket{\tilde{\psi}_{\bm{\alpha}}}/\sqrt{P_{\bm{\alpha}}}$.
From the linearity of the trace, we can write the average of $f$
as 
\begin{equation}
\begin{aligned}\mathbb{E}_{\ket{\psi}\sim\mathrm{Haar}(d^{N_{A}+N_{B}})}f(\ket{\psi}) & =\bra{i}\mathrm{Tr}_{B_{2}}\left(\sum_{\bm{\alpha}}\mathbb{E}_{\ket{\psi}\sim\mathrm{Haar}(d^{N_{A}+N_{B}})}\left(P_{\bm{\alpha}}(\ket{\psi_{\bm{\alpha}}}\bra{\psi_{\bm{\alpha}}})^{\otimes k}\right)\right)\ket{j}.\end{aligned}
\end{equation}
The Ref. \citep{cotler2023emergent} proved that If $\ket{\psi}$
is Haar-random distributed, $P_{\bm{\alpha}}$and $\ket{\psi_{\bm{\alpha}}}$would
be independent. Therefore, we can perform the average as 
\begin{equation}
\begin{aligned} & \mathbb{E}_{\ket{\psi}\sim\mathrm{Haar}(d^{N_{A}+N_{B}})}f(\ket{\psi})\\
 & =\sum_{\bm{\alpha}}\frac{1}{d^{N_{B_{1}}+N_{B_{3}}}}\bra{i}\mathrm{Tr}_{B_{2}}\left(\mathbb{E}_{\ket{\varphi}\sim\mathrm{Haar}(d^{N_{A}+N_{B_{2}}})}(\ket{\varphi}\bra{\varphi})^{\otimes k}\right)\ket{j},\\
 & =\bra{i}\rho_{\mathrm{GHS}}^{(k)}\ket{j},
\end{aligned}
\end{equation}
where in the first equality, we used that $\mathbb{E}_{\ket{\psi}\sim\mathrm{Haar}(d^{N_{A}+N_{B}})}(P_{\bm{\alpha}})=1/d^{N_{B_{1}}+N_{B_{3}}}$,
and in the second equality, we used the definition of the generalized
Hilbert-Schmidt ensemble. 

\end{proof}

Denoting the $k$-th moment of the MSPE corresponding to $\ket{\psi}$
as
\begin{equation}
\rho_{\ket{\psi}}^{(k)}=\sum_{\bm{\alpha}}P_{\bm{\alpha}}\mathrm{Tr}_{B_{2}}\left(\ket{\psi_{\bm{\alpha}}}\bra{\psi_{\bm{\alpha}}}\right)^{\otimes k},
\end{equation}
we can combine the above three lemmas to derive that
\begin{equation}
\mathrm{Prob}_{\ket{\psi}\sim\mathrm{Haar}(d^{N_{A}+N_{B}})}\left[\bra{i}\rho_{\ket{\psi}}^{(k)}\ket{j}-\bra{i}\rho_{\mathrm{GHS}}^{(k)}\ket{j}|\geq\frac{\epsilon}{d^{2N_{A}k}}\right]\leq2\exp\left(-\frac{d^{N_{A}+N_{B}}\epsilon^{2}}{18\pi^{3}(2k-1)^{2}d^{4N_{A}k}}\right),
\end{equation}
for $\forall\ket{i},\ket{j}$ the standard basis of $\mathcal{H}_{A}^{\otimes k}$.
Here, we have rescaled the error $\epsilon$ in Lemma \ref{Lemma:Levy}
to $\epsilon/d^{2N_{A}k}$ and substituted the dimension of total
Hilbert space as $D=d^{N_{A}+N_{B}}$. It is well known that 
\begin{equation}
\lVert O\rVert_{1}\leq\sum_{i,j}|O_{i,j}|
\end{equation}
for any matrix $O$. Therefore, the above inequality directly implies
that 
\begin{equation}
\mathrm{Prob}_{\ket{\psi}\sim\mathrm{Haar}(d^{N_{A}+N_{B}})}\left[\lVert\rho_{\ket{\psi}}^{(k)}-\rho_{\mathrm{GHS}}^{(k)}\rVert_{1}\geq\epsilon\right]\leq2d^{2N_{A}k}\exp\left(-\frac{d^{N_{A}+N_{B}}\epsilon^{2}}{18\pi^{3}(2k-1)^{2}d^{4N_{A}k}}\right),
\end{equation}
or equivalently
\begin{equation}
\mathrm{Prob}_{\ket{\psi}\sim\mathrm{Haar}(d^{N_{A}+N_{B}})}\left[\Delta_{1}^{(k)}\geq\epsilon\right]\leq2d^{2N_{A}k}\exp\left(-\frac{d^{N_{A}+N_{B}}\epsilon^{2}}{18\pi^{3}(2k-1)^{2}d^{4N_{A}k}}\right),\label{eq:SM_probability_equation_Delta_1}
\end{equation}
since $\lVert\rho_{\mathrm{GHS}}^{(k)}\rVert_{1}=1$. We can require
that
\begin{equation}
N_{B}\geq\frac{\log\left(\frac{18\pi^{3}(2k-1)^{2}d^{4N_{A}k}}{\epsilon^{2}}\log\frac{2d^{2N_{A}k}}{\delta}\right)}{\log d}=\Omega(N_{A}k+\log\frac{1}{\epsilon}+\log\log\frac{1}{\delta})
\end{equation}
such that the right hand side of Eq. (\ref{eq:SM_probability_equation_Delta_1})
is less than $\delta$, which proves the Theorem \ref{Theorem:global_haar}.

\subsection{Numerical results on typical models}
In this subsection, we report numerical calculations of different typical models at finite sizes, whose MSPE at late time all deeply thermalize to the generalized Hilbert-Schmidt ensemble, confirming the universality of this deep thermalization.

We consider four models: i) The local random Haar circuit; ii) The random dual-unitary circuit; iii) The kicked-Ising models away from dual-unitary points; iv) Mixed field Ising Hamiltonian dynamics. The model i)-iii) has exactly the same structure as the one shown in Fig.~\ref{fig1} in the main text, while the model iv) has a similar structure with the difference that the initial state is $\ket{0}^{\otimes N}$ and the measurement is performed in the computational basis. In all the models, the local Hilbert-space dimension is $d=2$, and we choose the leftmost two sites as the subsystem $A$ and the two sites in the middle as the subpart $B_2$.

Our numerical results are presented in Fig.~\ref{fig_SM}. As shown in the figure, all models we studied exhibit an exponential decay of the distance $\Delta_1^{(2)}$ from the generalized Hilbert–Schmidt ensemble as the system size $N$ increases. We further verified that this exponentially-decay behavior persists for $k=3$ and $k=4$, as well as for different measurement bases and initial states. These findings demonstrate the universality of the mixed-state deep thermalization proposed in this work.

Interestingly, the approach to saturation shows model-dependent features. For Hamiltonian dynamics, the distance $\Delta_1^{(2)}$ decreases polynomially with the evolution time $t$ prior to saturation, whereas for the other three circuit dynamics it decays exponentially in $t$. This distinction arises from whether or not energy is conserved in the dynamics. Moreover, random dual-unitary circuits display a two-stage deep thermalization process: for $t<N$, $\Delta_1^{(2)}$ converges more rapidly to a smaller value, while for $t>N$ it settles at a slightly higher saturation value. This behavior agrees with the $k=1$ case of conventional thermalization. At $k=1$ and for $t<N$, dual-unitary circuits thermalize exactly and instantaneously~\cite{bertini2019exact, piroli2020exact}, whereas for $t>N$ this exact thermalization mechanism breaks down due to the effect of the boundaries~\cite{suzuki2022computational, PhysRevLett.133.190401}. A detailed exploration of these early-time behaviors across different models is left for future work.

\begin{figure}
\includegraphics[width=0.7\columnwidth]{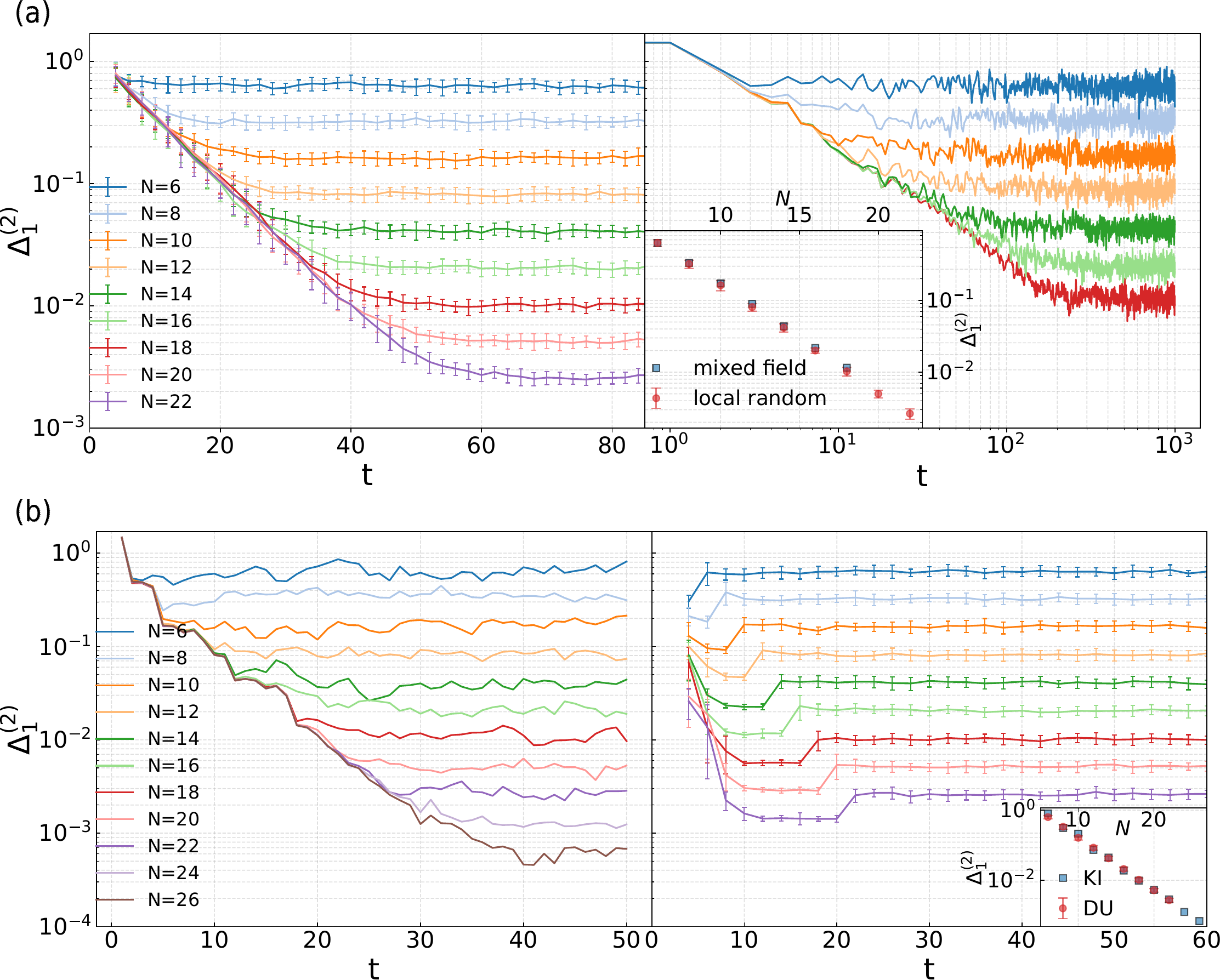}
\caption{
Distance of $\rho_{N_{A}}^{(k)}$ from the generalized Hilbert-Schmidt ensemble. (a) Repeated from Fig.~\ref{fig2} in the main text. Left: local Haar-random circuits for system size $N$ from $6$ to $22$, where error bars indicate statistical fluctuations over different circuit realizations. Right: mixed-field Ising dynamics for $N$ from $6$ to $18$, governed by the Hamiltonian $H=\sum_{j=1}^N(h_x\sigma_x^j+h_y\sigma^j_y)+\sum_{j=1}^{N-1}J\sigma_x^j\sigma_x^{j+1}$, with $\sigma_\alpha^j$ the Pauli matrix on $j$-th qubit and $(h_x,h_y,J)=(0.8090,0.9045,1.0)$.
Inset: saturated deviation $\Delta_1^{(2)}$ at late times versus the system size; for the mixed-field Ising case, $\Delta_1^{(2)}$ is averaged over time after saturation.
(b) Left: The kicked-Ising model where each circuit layer corresponds to the unitary $U=\exp(-ih\sum_{j=1}^N\sigma_y^j)\exp(-iJ\sum_{j=1}^{N-1}\sigma_z^j\sigma_z^{j+1})\exp(-ig\sum_{j=1}^N\sigma^j_z)$ with $(h, J, g)=(0.9, 0.7, 0.6)$ away from the dual-unitary points. 
(b) Right: The random dual-unitary circuits where each gate in the brick-wall structure is a randomly chosen dual-unitary gate. 
Error bars indicate statistical fluctuations from $25$ different  circuit realizations. Inset: saturated deviation $\Delta_1^{(2)}$ at late times versus the system size for the kicked-Ising model (KI) and the random dual-unitary circuits (DU).
}
\label{fig_SM}
\end{figure}

\section{conditional entropy}
\label{Sec:ConditionalInformation}
In this section, we compute the annealed average Renyi-$k$ conditional entropy introduced in Eq. (\ref{eq:conditional_entropy_expression}), which requires multiple copies of the system. In the following, we consider the asymptotic expression for an arbitrary integer $k$ under the limit $t\to\infty$ and $N_A, m\to\infty$ while their ratio $\gamma=m/N_A$ is fixed.
We use the symbol $e$ to denote the identity permutation and $\sigma'_k$ to denote the mapping from $\{1,2,3,\cdots,k\}$ to $\{k,1,2,\cdots,k-1\}$.
In addition,
the density matrix with an argument, like $\rho_{S,k}(g)$, always denotes the representation of a permutation element $g$ on $k$-replica with support on sites in region $S$ ($S=A$ or $S=R$) in each replica. The definition is almost the same as the one in Eq.~(\ref{eq:SM_specific_representation}), except that here we denote the region instead of the sizes of the region in the subscript.
It should not be confused with the reduced density matrix on a single replica, like $\rho_A$ and $\rho_{AR}$.

Computation of $S_{AR}^{(k)}$ requires us to calculate $\mathrm{Tr}\overline{\rho_{AR}^{k}}=\sum_{\bm{\alpha}}P_{\bm{\alpha}}\mathrm{Tr}\rho_{AR}^k(\bm{\alpha})$. Similarly as Eq.~(\ref{eq:SM_expression_of_distribution_EV}), we have
\begin{equation}
\mathrm{Tr}\overline{\rho_{AR}^k}=\mathrm{Tr}\left((\rho_{A,k}^\dagger(\sigma_k')\otimes\rho_{R,k}^\dagger(\sigma_k'))\rho_{AR}^{(k)}\right),
\end{equation}
with $\rho_{AR}^{(k)}=\overline{\rho_{AR}^{\otimes k}}$ being the $k$-th moment of MSPE \footnote{Here the reference part $R$ is also an output}.

To proceed, we also use the replica trick to compute $\mathrm{Tr}\overline{\rho_{AR}^k}$. We define
\begin{equation}
\mathrm{Tr}_q\overline{\rho_{AR}^{k}}=\mathrm{Tr}\left((\rho_{A,k}^{\dagger}(\sigma_{k}')\otimes\rho_{R,k}^{\dagger}(\sigma_{k}'))\overline{\rho}_{AR}^{(k,q)}\right)
\end{equation}
with $\overline{\rho}_{AR}^{(k,q)}$ defined in Eq. (\ref{SM:expression_of_bar_rho_k_q}). In the limit $q=1-k$, $\mathrm{Tr}_q\overline{\rho_{AR}^{k}}$ gives rise to $\mathrm{Tr}\overline{\rho_{AR}^k}$. 
Similarly, we define 
\begin{equation}
\mathrm{Tr}_q\overline{\rho_{A}^k}=\mathrm{Tr}\left((\rho_{A,k}^\dagger(\sigma_k')\otimes\rho_{R,k}^\dagger(e))\overline\rho_{AR}^{(k,q)}\right).
\end{equation}
We will compute the quantity
\begin{equation}
I_{R:A}^{(k,q)}=-\frac{1}{k-1}\log\frac{\mathrm{Tr}_q\overline{\rho_{AR}^{k}}}{\mathrm{Tr}_q\overline{\rho_{A}^{k}}}
\end{equation}
for any integer $q>0$, and then take the limit $q=1-k$ which gives us the conditional entropy $I_{R:A}^{(k)}$ defined in Eq. (\ref{eq:conditional_entropy_expression}).

$\mathrm{Tr}_q\overline{\rho_{AR}^{k}}$ can be re-expressed as 
\begin{equation}
\mathrm{Tr}_{q}\overline{\rho_{AR}^{k}}=\mathrm{Tr}\left((\rho_{A,k+q}^{\dagger}(\sigma)\otimes\rho_{R,k+q}^{\dagger}(\sigma))\tilde{\rho}_{AR}^{(k+q)}\right)
\end{equation}
with $\sigma=(\sigma'_k, e)$. Notice that $\sigma'_k$ is a permutation on $k$ elements and $\sigma$ is a permutation on $k+q$ elements.
It is diagrammatically represented as
\begin{equation*}
\mathrm{Tr}_q\overline{\rho_{AR}^k}=
\begin{tikzpicture}[baseline=(current  bounding  box.center), scale=0.5]
\foreach \i in {0,2,5,7,9,12}
{
\foreach \j in {0,2}
{
\Wgategreenprime{\i}{\j}
}
}
\foreach \i in {-1,1,4,6,8,10,13}
{
\foreach \j in{1,3}
{
\Wgategreenprime{\i}{\j}
}
}
\Text[x=3.1,y=2]{$\bm{\cdots}$}

\Text[x=10.9,y=2]{$\bm{\cdots}$}

\foreach \i in {0,2,9,12}
{
\Measurement{\i}{4}
}
\foreach \i in{5,7}
{
\MYcircle{\i-0.5+0.1}{3.5+0.1}
}
\foreach \i in {6,8}
{
\MYcircle{\i-0.5-0.1}{3.5+0.1}
}
\foreach \i in {-2,0,3,5,7,9,12}
{
\draw[very thick]
  (\i+0.5,-0.5) to[out=-45, in=225] (\i+1.5,-0.5);
}
\foreach \i in {-1,1}
{
\draw[very thick]
  (13.5,\i+0.5) to[out=45, in=315] (13.5,\i+1.5);
}
\draw[very thick] (13.5,3.5)--(13.5,4);
\foreach \i in {-1,1}
{
\draw[very thick]
  (-1.5,\i+0.5) to[out=135, in=225] (-1.5,\i+1.5);
}
\draw[very thick] (-1.5,3.5) -- (-1.5,4);
\draw [decorate, decoration={brace, amplitude=10pt}, very thick] (0,4.5) -- (3.5,4.5);
\draw [decorate, decoration={brace, amplitude=10pt}, very thick] (4.5,4.5) -- (7.5,4.5);
\draw [decorate, decoration={brace, amplitude=10pt}, very thick] (8.5,4.5) -- (12.5,4.5);
\draw [decorate, decoration={brace, amplitude=10pt}, very thick] (-1.5,4.5) -- (-0.5,4.5);
\Text[x=-1,y=5.75]{\large{$\bm{A}$}}
\Text[x=1.75, y=5.75]{\large{$\bm{B_3}$}}
\Text[x=6, y=5.75]{\large{$\bm{B_2}$}}
\Text[x=10.5, y=5.75]{\large{$\bm{B_1}$}}
\Text[x=13.5,y=5.75]{\large{$\bm{R}$}}
\draw[->,very thick](14.5,-1.5)--(14.5,6);
\Text[x=15,y=2.5] {\pmb{t}}
\MYsquare{-1.5}{4.1}
\MYsquare{13.5}{4.1}
\draw[very thick, dashed,color=blue](-0.45,4.0)--(-0.45,-1);
\Text[x=-0.45,y=-2]{\large$\bm{\rho_\mathrm{f}}$}
\draw[very thick,dashed,color=blue] (4.35,3.65)--(4.35,-1);
\Text[x=4.35,y=-1.5]{\large$\bm{\rho_\mathrm{a}}$}
\draw[very thick,dashed,color=blue] (7.65,3.65)--(7.65,-1);
\Text[x=7.65,y=-1.5]{\large$\bm{\rho_\mathrm{b}}$}
\draw[very thick, dashed,color=blue](13.45,4.0)--(13.45,-1);
\Text[x=13.45,y=-1.5]{\large$\bm{\rho_\mathrm{i}}$}
\draw[->,very thick](13.4,-2.5)--(4,-2.5);
\Text[x=8.7,y=-3]{Horizontal Evolution}
\end{tikzpicture}.
\end{equation*}
Here, $k'=k+q$, and we use the white square to represent the vectorized state $\rho(\sigma)$.

We follow a similar calculation in the horizontal spatial direction as in the Sec. \ref{Sec:k_th_moment_calculation}.
The rightmost boundary state $\rho_i$ can be represented  as
\begin{equation}
\rho_{i}=\bigg(\ket{\varphi_{0}}^{\otimes\frac{t}{2}}\bra{\varphi_{0}}^{\otimes\frac{t}{2}}\bigg)^{\otimes (k+q)}\otimes\rho_{R,k+q}(\sigma),
\end{equation}
with $\ket{\varphi_0}=\sum_{i=1}^d\ket{ii}$.
Suppose the steady state before the region $B_2$, $\rho_\mathrm{b}$, can be expressed
as 
\begin{equation}
\rho_\mathrm{b}=\sum_{g\in\mathbb{S}_{k+q}}\alpha(g)\rho_{t+1,k+q}(g).\label{eq:SM_sec6_rhof_AR_krenyi}
\end{equation}
The coefficients in front of different representations, $\alpha(g)$, are determined from the invariance of the overlap of $\rho_i$ and $\rho_b$ with $\rho_{t+1,k+q}(g)$. 
Since $\mathrm{Tr}\rho_{t,k+q}^\dagger(g)\bigg(\ket{\varphi_{0}}^{\otimes\frac{t}{2}}\bra{\varphi_{0}}^{\otimes\frac{t}{2}}\bigg)^{\otimes (k+q)}=d^{\frac{t(k+q)}{2}}$, 
we have
\begin{equation}
\sum_{g'\in\mathbb{S}_{k+q}}\alpha(g')d^{(t+1)l(g^{-1}g')}=d^{l(g^{-1}\sigma)}d^{\frac{t(k+q)}{2}}\ \mathrm{for}\ \forall g\in\mathbb{S}_{k+q}
.\label{eq:SM_sec6_rhof_AR_result_krenyi}
\end{equation}
In the limit $t\to\infty$ the leading term is for $g'=g$, so we can solve the equations, obtaining 
\begin{equation}
\alpha(g)=d^{l(g^{-1}\sigma)}d^{-\frac{t(k+q)}{2}-(k+q)}\label{eq:SM_sec6_krenyi_result_alpha}.
\end{equation}

The subsystems $B_2$ and $B_3$ mix $\rho_{t+1}(g)$ with different $g$, leading to different coefficients in the state $\rho_\mathrm{f}$. Suppose 
\begin{equation}
\rho_\mathrm{f}=\sum_{g\in\mathbb{S}_{k+q}} \beta(g)\rho_{t+1,k+q}(g).
\end{equation}
They can be related to $\alpha(g)$ similarly as in Sec.~\ref{Sec:k_th_moment_calculation} by computing the overlap of $\rho_{t+1,k+q}(g)$ with $\rho_\mathrm{f}$ and $\rho_\mathrm{a}$ respectively, resulting in
\begin{equation}
\sum_{g'\in\mathbb{S}_{k+q}}\beta(g')d^{(t+1)l(g^{-1}g')}=\sum_{g'\in\mathbb{S}_{k+q}}\alpha(g')\frac{d^{(t+1-m/2)l(g^{-1}g')+m/2[l(g)+l(g')]}}{d^{m(k+q)/2}}.\label{eq:SM_CE_relation_between_betaalpha}
\end{equation}
In the limit $t\to\infty$, these equations  simplify to
\begin{equation}
\beta(g)=\alpha(g)d^{m[l(g)-(k+q)]}.\label{eq:new_transformation_rule_krenyi}
\end{equation}

Combining Eqs. (\ref{eq:SM_sec6_krenyi_result_alpha}) and (\ref{eq:new_transformation_rule_krenyi}), we  express $\rho_\mathrm{f}$ as
\begin{equation}
\begin{aligned}
\rho_\mathrm{f} =\sum_{g\in\mathbb{S}_{k+q}}d^{-\frac{t(k+q)}{2}-(k+q)}d^{l(g^{-1}\sigma)}d^{m[l(g)-(k+q)]}\rho_{t+1,k+q}(g).
\end{aligned}
\end{equation}
The gates in the subsystem $A$ form an isometry, which does not change the coefficients $\beta(g)$ as shown in the Sec.~\ref{Sec:k_th_moment_calculation}. Therefore, we can compute 
$-1/(k-1)\log\mathrm{Tr}_q\overline{\rho_{AR}^k}$ as
\begin{equation}
-\frac{\log\mathrm{Tr}_q\overline{\rho_{AR}^k}}{k-1}=-\frac{1}{k-1}\log\sum_{g\in\mathbb{S}_{k+q}}\left( d^{-\frac{t(k+q)}{2}-(k+q)}d^{l(g^{-1}\sigma)}d^{m[l(g)-(k+q)]}d^{N_{A}l(\sigma^{-1}g)}\right).
\end{equation}

Next, we will compute $\mathrm{Tr}_q\overline{\rho_{A}^k}$. Since
\begin{equation}
\mathrm{Tr}_q\overline{\rho_{A}^k}=\mathrm{Tr}\left((\rho_{A,k+q}^\dagger(\sigma)\otimes\rho_{R,k+q}^\dagger(e))\tilde\rho_{AR}^{(k+q)}\right)
\end{equation}
with $\rho_{R,k+q}(e)$ being the Identity operator on the region $R$ across $k+q$ replicas,
it can be diagrammatically represented as 
\begin{equation*}
\mathrm{Tr}_q\overline{\rho_{A}^k}=
\begin{tikzpicture}[baseline=(current  bounding  box.center), scale=0.5]
\foreach \i in {0,2,5,7,9,12}
{
\foreach \j in {0,2}
{
\Wgategreenprime{\i}{\j}
}
}
\foreach \i in {-1,1,4,6,8,10,13}
{
\foreach \j in{1,3}
{
\Wgategreenprime{\i}{\j}
}
}
\Text[x=3.1,y=2]{$\bm{\cdots}$}

\Text[x=10.9,y=2]{$\bm{\cdots}$}

\foreach \i in {0,2,9,12}
{
\Measurement{\i}{4}
}
\foreach \i in{5,7}
{
\MYcircle{\i-0.5+0.1}{3.5+0.1}
}
\foreach \i in {6,8}
{
\MYcircle{\i-0.5-0.1}{3.5+0.1}
}
\foreach \i in {-2,0,3,5,7,9,12}
{
\draw[very thick]
  (\i+0.5,-0.5) to[out=-45, in=225] (\i+1.5,-0.5);
}
\foreach \i in {-1,1}
{
\draw[very thick]
  (13.5,\i+0.5) to[out=45, in=315] (13.5,\i+1.5);
}
\draw[very thick] (13.5,3.5)--(13.5,4);
\foreach \i in {-1,1}
{
\draw[very thick]
  (-1.5,\i+0.5) to[out=135, in=225] (-1.5,\i+1.5);
}
\draw[very thick] (-1.5,3.5) -- (-1.5,4);
\draw [decorate, decoration={brace, amplitude=10pt}, very thick] (0,4.5) -- (3.5,4.5);
\draw [decorate, decoration={brace, amplitude=10pt}, very thick] (4.5,4.5) -- (7.5,4.5);
\draw [decorate, decoration={brace, amplitude=10pt}, very thick] (8.5,4.5) -- (12.5,4.5);
\draw [decorate, decoration={brace, amplitude=10pt}, very thick] (-1.5,4.5) -- (-0.5,4.5);
\Text[x=-1,y=5.75]{\large{$\bm{A}$}}
\Text[x=1.75, y=5.75]{\large{$\bm{B_3}$}}
\Text[x=6, y=5.75]{\large{$\bm{B_2}$}}
\Text[x=10.5, y=5.75]{\large{$\bm{B_1}$}}
\Text[x=13.5,y=5.75]{\large{$\bm{R}$}}
\draw[->,very thick](14.5,-1.5)--(14.5,6);
\Text[x=15,y=2.5] {\pmb{t}}
\MYsquare{-1.5}{4.1}
\MYcircle{13.5}{4.1}
\draw[very thick, dashed,color=blue](-0.45,4.0)--(-0.45,-1);
\Text[x=-0.45,y=-2]{\large$\bm{\rho_\mathrm{f}}$}
\draw[very thick,dashed,color=blue] (4.35,3.65)--(4.35,-1);
\Text[x=4.35,y=-1.5]{\large$\bm{\rho_\mathrm{a}}$}
\draw[very thick,dashed,color=blue] (7.65,3.65)--(7.65,-1);
\Text[x=7.65,y=-1.5]{\large$\bm{\rho_\mathrm{b}}$}
\draw[very thick, dashed,color=blue](13.45,4.0)--(13.45,-1);
\Text[x=13.45,y=-1.5]{\large$\bm{\rho_\mathrm{i}}$}
\draw[->,very thick](13.4,-2)--(4,-2);
\Text[x=8.7,y=-2.5]{Horizontal Evolution}
\end{tikzpicture}.
\end{equation*}
Here, $k'=k+q$, and we use the empty bullet to represent the vectorized Identity operator $\rho(e)$.

Now, the rightmost boundary state can be expressed as 
$\rho_i=\bigg(\ket{\varphi_{0}}^{\otimes\frac{t}{2}}\bra{\varphi_{0}}^{\otimes\frac{t}{2}}\bigg)^{\otimes (k+q)}\otimes\rho_{R,k+q}(e)$.
Suppose the $\rho_b$ is 
\begin{equation}
\rho_{\mathrm{b}}=\sum_{g\in\mathbb{S}_{k+q}} \alpha(g) \rho_{t+1,k+q}(g),\label{eq:SM_Sec6_rhof_A}
\end{equation}
which satisfies the equation
\begin{equation}
\sum_{g'\in\mathbb{S}_{k+q}}\alpha(g')d^{(t+1)l(g^{-1}g')}=d^{l(g^{-1})}d^{\frac{t(k+q)}{2}}\ \mathrm{for}\ \forall g\in\mathbb{S}_{k+q}
\label{eq:SM_sec6_rhof_A_result_krenyi}
\end{equation}
by equating the overlap of $\rho_b$ and $\rho_i$ with $\rho_{t+1,k+q}(g)$.
In the limit $t\to\infty$, we have
\begin{equation}
\alpha(g)=d^{l(g^{-1})}d^{-\frac{t(k+q)}{2}-(k+q)}.
\end{equation}

The effect of the subsystem $B_2$, $B_3$ can be expressed by the linear transformation rule Eq. (\ref{eq:new_transformation_rule_krenyi}). Consequently,
the state $\rho_\mathrm{f}$ is written as
\begin{equation}
\begin{aligned}
\rho_\mathrm{f} = \sum_{g\in\mathbb{S}_{k+q}}d^{-\frac{t(k+q)}{2}-(k+q)}d^{l(g^{-1})}d^{m[l(g)-(k+q)]}\rho_{t+1,k+q}(g).
\end{aligned}
\end{equation}
$-1/(k-1)\log\mathrm{Tr}_q\overline{\rho_{A}^k}$ is then obtained as
\begin{equation}
-\frac{\log\mathrm{Tr}_q\overline{\rho_{A}^k}}{k-1}=-\frac{1}{k-1}\log\sum_{g\in\mathbb{S}_{k+q}}\left(d^{-\frac{t(k+q)}{2}-(k+q)}d^{l(g)}d^{m[l(g)-(k+q)]}d^{N_{A}l(\sigma^{-1}g)}\right).
\end{equation}

Therefore, we have the result of $I_{R:A}^{(k,q)}$ in the presence of $m$ sites with forgotten measurements
\begin{equation}
\begin{aligned}
I_{R:A}^{(k,q)} =-\frac{1}{k-1}\log\frac{\sum_{g\in\mathbb{S}_{k+q}}d^{l(g^{-1}\sigma)}d^{ml(g)}d^{N_{A}l(\sigma^{-1}g)}}{\sum_{g\in\mathbb{S}_{k+q}}d^{l(g)}d^{ml(g)}d^{N_{A}l(\sigma^{-1}g)}}.\label{eq:SM_conditional_withnoise_krenyi}
\end{aligned}
\end{equation}
Before proceeding, we first introduce a mathematical inequality 
\begin{equation}
l(g)+l(\sigma^{-1}g)\leq k+1+2q.\label{eq:SM_loopinequality}
\end{equation}
This inequality can be derived from the Caylay distance~\cite{delima2012cayley}, which is defined as $\mathrm{Dis}(g_1,g_2)=k+q-l(g_1^{-1}g_2)$. By the triangle inequality, we have
\begin{equation}
k+q-l(\sigma)=\mathrm{Dis}(g^{-1},g^{-1}\sigma)\leq\mathrm{Dis}(g^{-1},e)+\mathrm{Dis}(e,g^{-1}\sigma)=k+q-l(g)+k+q-l(g^{-1}\sigma).
\end{equation}
The desired inequality is obtained after reorganizing and using $l(\sigma)=q+1$ and $l(g^{-1}\sigma)=l(\sigma^{-1}g)$. Notice that the equality can be achieved, for example, when $g=e$ and $g=\sigma$.

Now, we can further simplify Eq. (\ref{eq:SM_conditional_withnoise_krenyi}) under the assumption that $N_A,m\to\infty$ with their ratio $\gamma=m/N_A$ fixed or alternatively $d\to\infty$. More concretely, 
we consider the exponential appearing both in the numerator and denomenator, $d^{ml(g)}d^{N_Al(\sigma^{-1}g)}$. Its exponent can be written as $N_A\left(\gamma l(g)+l(\sigma^{-1}g\right)$. If $\gamma<1$, sing Eq.~(\ref{eq:SM_loopinequality}), we have
\begin{equation}
\gamma l(g)+l(\sigma^{-1}g)\leq\gamma\left(k+2q+1-l(\sigma^{-1}g)\right)+l(\sigma^{-1}g)\leq k+q+\gamma(q+1).
\end{equation}
The maximum is achieved only if $g=\sigma$.
Similarly, if $\gamma >1$, we have
\begin{equation}
\gamma l(g)+l(\sigma^{-1}g)\leq\gamma l(g)+k+2q+1-l(g)\leq\gamma(k+q)+q+1.
\end{equation}
The maximum is achieve only if $g=e$. In summary, we have the following conclusions:

\begin{enumerate}
\item If $N_{A}>m$, the main contribution in Eq.~(\ref{eq:SM_conditional_withnoise_krenyi}) is from $g=\sigma$, which gives us $I_{R:A}^{(k,q)}=-\log(d)$;
\item $N_{A}<m$, the main contribution in Eq.~(\ref{eq:SM_conditional_withnoise_krenyi}) is from $g=e$, which gives us $I_{R:A}^{(k,q)}=\log(d)$;
\item If $N_{A}=m$, $I_{R:A}^{(k,q)}=0$ no matter how large $d$ or $N_A$ is. This can be seen from Eq.~(\ref{eq:SM_conditional_withnoise_krenyi}) by changing $g\to g^{-1}\sigma$ in the denomenator.
\end{enumerate}
Since the above result is independent of $q$, we can take the analytic continuation to  $q=1-k$ and expect that the same result holds for $I_{R:A}^{(k)}$. This concludes the sharp phase transition of $I_{R:A}^{(k)}$ at $N_A=m$ for every $k$.

The above conclusion can also be reduced to $m=0$, i.e., without any forgottened measurements. In this case, as long as $N_A\to\infty$ or equivalently $d\to\infty, N_A>1$, we always have $I_{R:A}^{(k)}=-\log(d)$.

\section{Deep thermalization with $m/2$ sparse corrupted  measurement outcomes}
\label{Sec:Modelwithsparseerausureerrors}
\begin{figure}[t]
\centering
\begin{tikzpicture}[baseline=(current  bounding  box.center), scale=0.5]
\foreach \i in {0,2,5,7,11,13,16}
{
\foreach \j in {0,2,4}
{
\Wgategreen{\i}{\j}
}
}
\foreach \i in {-1,1,4,6,8,10,12,14,17}
{
\foreach \j in{1,3}
{
\Wgategreen{\i}{\j}
}
}
\Text[x=3.1,y=0]{$\bm{\cdots}$}
\Text[x=3.1,y=4]{$\bm{\cdots}$}
\Text[x=9.1,y=0]{$\bm{\cdots}$}
\Text[x=9.1,y=4]{$\bm{\cdots}$}
\Text[x=15.1,y=0]{$\bm{\cdots}$}
\Text[x=15.1,y=4]{$\bm{\cdots}$}
\foreach \i in {1,4,8,10,14,17}
{
\Measurement{\i}{5}
}
\foreach \i in{5,11}
{
\MYcircle{\i+0.5+0.1}{4.5+0.1}
}
\foreach \i in {7,13}
{
\MYcircle{\i-0.5-0.1}{4.5+0.1}
}
\foreach \i in {-2,0,3,5,7,9,11,13,16}
{
\draw[very thick]
  (\i+0.5,-0.5) to[out=-45, in=225] (\i+1.5,-0.5);
}
\foreach \i in {-1,1,3}
{
\draw[very thick]
  (17.5,\i+0.5) to[out=45, in=315] (17.5,\i+1.5);
}
\foreach \i in {-1,1,3}
{
\draw[very thick]
  (-1.5,\i+0.5) to[out=135, in=225] (-1.5,\i+1.5);
}
\draw[very thick,dashed](0.4,-1.5)--(0.4,6);
\Text[x=-1,y=5.5]{\large{$\bm{\tilde{\rho}_{N_A}^{(k)}}$}}
\draw[->,very thick](18,-1.5)--(18,6);
\Text[x=18.5,y=2.5] {\pmb{t}}
\Text[x=0.4,y=-2]{\large$\bm{\rho_n}$}
\draw[very thick, dashed,color=blue](17.45,5.0)--(17.45,-1);
\Text[x=17.45,y=-1.5]{\large$\bm{{\rho'}_0}$}
\draw[very thick, dashed, color=blue](11.4,5)--(11.4,-1);
\Text[x=11.4,y=-1.5]{\large$\bm{{\rho'}_1}$}
\draw[very thick, dashed, color=blue](12.6,5)--(12.6,-1);
\Text[x=12.6,y=-1.5]{\large$\bm{\rho_0}$}
\draw[very thick, dashed, color=blue](6.6,5)--(6.6,-1);
\draw[very thick, dashed, color=blue](5.4,5)--(5.4,-1);
\Text[x=5.4,y=-1.5]{\large$\bm{{\rho'}_{i+1}}$}
\Text[x=6.6,y=-1.5]{\large$\bm{\rho_i}$}
\draw[->,very thick](17.4,-2)--(4,-2);
\Text[x=8.7,y=-2.5]{Horizontal Evolution}
\draw [decorate, decoration={brace, mirror, amplitude=10pt}, very thick] (17,5.5) -- (13,5.5);
\Text[x=15,y=6.75]{\large{$\bm{B_1}$}}
\draw [decorate, decoration={brace, mirror, amplitude=10pt}, very thick] (11.25,5.5) -- (9.25,5.5);
\Text[x=10.25,y=6.75]{\large{$\bm{B_2}$}}
\draw [decorate, decoration={brace, mirror, amplitude=10pt}, very thick] (8.75,5.5) -- (6.75,5.5);
\Text[x=7.75,y=6.75]{\large{$\bm{B_{i+1}}$}}
\draw [decorate, decoration={brace, mirror, amplitude=10pt}, very thick] (5.25,5.5) -- (3.25,5.5);
\Text[x=4.25,y=6.75]{\large{$\bm{B_{i+2}}$}}
\draw [decorate, decoration={brace, mirror, amplitude=10pt}, very thick] (2.75,5.5) -- (0.75,5.5);
\Text[x=1.75,y=6.75]{\large{$\bm{B_{n+1}}$}}
\end{tikzpicture}
\caption{Quench dynamics leading to MSPE with  $m/2$ sparse corrupted  measurement outcomes. }
\label{Fig:SM_sparseerror}
\end{figure}

In this section, we focus on $m/2$ sparse corrupted  measurement outcomes on $n$ sites ($m=2n$) and show that the system deep thermalizes to the generalized Hilbert-Schmidt ensemble also in this case. We consider the spatial evolution as illustrated in Fig.~\ref{Fig:SM_sparseerror}, where ${\rho'}_0, \rho_n$ correspond to the initial and final state, i.e. $\rho_\mathrm{i}$ and $\rho_\mathrm{f}$, in Fig.~\ref{fig:kmomentSM}, respectively.

We assume that the circuit between each two pairs of erasure errors, i.e., the circuit $B_i$, is long enough such that the time-like state $\rho_i$ has been projected onto the steady space of the spatial transfer matrix. Consequently, we can express $\rho_i$ as 
\begin{equation}
\rho_i=\sum_g \alpha_i(g) \rho_{t+1,k}(g).
\end{equation}
After the pair of erasure sites, the time-like state ${\rho'}_{i+1}$ becomes 
\begin{equation}
{\rho'}_{i+1}=\sum_{g}\frac{d^{l(g)}\alpha_{i}(g)\rho_{t,k}(g)\otimes\rho_{1,k}(e)}{d^{k}},
\end{equation}
which is further projected to the steady state 
$\rho_{i+1}=\sum_g \alpha_{i+1}(g) \rho_{t+1,k}(g)$.
Since the overlap of $\rho_{t+1,k}(g)$ with ${\rho'}_{i+1}$ is the same as that 
with $\rho_{i+1}$, we obtain
\begin{equation}
\sum_{g'}\alpha_{i+1}(g')\frac{d^{(t+1)l(g^{-1}g')}}{d^{(t+1)k}}=\sum_{g'}\alpha_{i}(g')\frac{d^{l(g')+l(g)+tl(g^{-1}g')}}{d^{k+(t+1)k}}\ \mathrm{for}\ \forall g. \label{Eq:SM_equation_determing_recursive}
\end{equation}
This recursive relation determines the state $\rho_n$ once we know the initial time-like state ${\rho'}_0$. At the limit $t\to\infty$, the above recursive relation can be simplified to $\alpha_{i+1}(g)=\alpha_i(g)d^{-2l(g)}$ and consequently
\begin{equation}
\alpha_{n}(g)=\alpha_0(g)d^{-2nl(g)}=\alpha_0(g)d^{-ml(g)}.\label{Eq:SM_recursive_relation_sparse_error}
\end{equation}
Since the initial state ${\rho'}_0$ is an equal superposition of all $\rho_{t+1}(g)$, we have $\alpha_0(g)=1$ up to a constant. Therefore,
\begin{equation}
\alpha_{n}(g)\propto d^{-ml(g)},
\end{equation}
which agrees with the results in Sec. \ref{Sec:large_t}. By performing the same analysis, we conclude that the MSPE for this setting is also deeply thermalized to the generalized Hilbert-Schmidt ensemble.
At the leading order of $d^t$, the Renyi-$k$ conditional entropy is also exactly the same as that derived in Sec.~\ref{Sec:ConditionalInformation}. This can be observed by noting that the transformation rule Eq.~(\ref{Eq:SM_recursive_relation_sparse_error}) is the same as Eq.~(\ref{eq:new_transformation_rule_krenyi}).
Nonetheless, in the next order of $\mathcal{O}\left(\frac{1}{d^t}\right)$, Eq.~(\ref{Eq:SM_equation_determing_recursive}) becomes non-diagonal, which will mix different permutation coefficients $\alpha(g)$. At finite $t$, this gives rise to a different result from the case with consecutive corrupted measurement outcomes studied in the main text.

\end{document}